\documentclass[a4paper,12pt]{extarticle}%{scrartcl}%{report}

\usepackage[tmargin=5.8cm]{geometry}

\usepackage{chngcntr}
\counterwithin{paragraph}{subsubsection}

\usepackage{cancel}%para tachado diagonal
\usepackage{graphicx}
\usepackage{amscd}
\usepackage{pb-diagram}

\usepackage{newclude}

\usepackage [english]{ babel } %\usepackage [ english , spanish ]{ babel }
\usepackage{float} %para poner tabla en lugar fijo
\usepackage[all]{xy} %crear un diagrama es mediante una matriz
\usepackage{amsmath,amsfonts,amssymb}
\usepackage{amstext} 
\usepackage{color}
\usepackage[latin1]{inputenc}
\usepackage{setspace}

\usepackage{pst-node}
\usepackage[urlcolor=blue,citecolor=blue,linkcolor=blue]{hyperref}
\usepackage{tocloft}

\usepackage{amsthm}
\newcommand{\listdefinitionsname}{\Large{List of Definitions}}
\newlistof{mydefinitions}{def}{\listdefinitionsname}
\newcommand{\mydefinitions}[1]{%
\addcontentsline{def}{mydefinitions}{\protect\numberline{\thedefinition}#1}\par}

\newcommand{\ws}{\textcolor{white}{e}}

\newtheorem{theorem}{Theorem}

\newtheorem{corollary}[theorem]{Corollary}

\newtheorem{definition}{Definition}

\newtheorem{lemma}[theorem]{Lemma}

\newtheorem{proposition}[theorem]{Proposition}
\newtheorem{remark}[theorem]{Remark}

\newtheoremstyle{named}{}{}{\itshape}{}{\bfseries}{.}{.5em}{\thmnote{#3} #2}
\theoremstyle{named} 
\newtheorem{namedtheorem}[theorem]{Theorem}

\newcommand{\enter} {\vskip 0.3cm}

\newcommand{\A}{\mathbb A}

\newcommand{\N}{\mathbb N}
\newcommand{\Z}{\mathbb Z}
\newcommand{\C}{\mathbb C}

\newcommand{\klk}{,\ldots,}
\newcommand{\ol}{\overline}
\newcommand{\MM}{\ol{\mathcal{M}}}

\newcommand{\M}{{\mathcal{M}}}

\DeclareMathAlphabet{\mathpzc}{OT1}{pzc}{m}{it} %para mathcal en minuscula
\begin{document}

\begin{center}
\textbf{\Large{On the intrinsic complexity of elimination problems in effective algebraic geometry
\footnote{
Research partially supported by the following Argentinian, Belgian and Spanish grants: CONICET PIP 2461/01, UBACYT 20020100100945, PICT--2010--0525, FWO G.0344.05, MTM2010-16051. 
}
%\footnote{Research partially supported by the following Argentinian, Belgian and Spanish grants: CONICET PIP 2461/01, UBACYT 20020100100945, PICT--2010--0525, FWO G.0344.05, MTM2010-16051. }
}}
\end{center}

\begin{center}
\Large{Joos Heintz\footnote{Departamento de Computaci\'on, Universidad de Buenos Aires and CONICET, Ciudad Universitaria, Pab. I, 1428 Buenos Aires, Argentina, and Departamento de Matem\'aticas, Estad\'istica y Computaci\'on, Facultad de Ciencias, Universidad de Cantabria, Avda. de los Castros s/n, 39005 Santander, Spain. joos@dc.uba.ar \& joos.heintz@unican.es}, Bart Kuijpers\footnote{Database and Theoretical Computer Science Research Group, Hasselt University, Agoralaan, Gebouw D, 3590 Diepenbeek, Belgium. bart.kuijpers@uhasselt.be}, Andr\'es Rojas Paredes\footnote{Departamento de Computaci\'on, Universidad de Buenos Aires, Ciudad Universitaria, Pab. I, 1428 Buenos Aires, Argentina. arojas@dc.uba.ar}\\}
\end{center}

\enter

\begin{abstract}
The representation of polynomials by arithmetic circuits evaluating them is an alternative data structure which allowed considerable progress in polynomial equation solving in the last fifteen years. We present a circuit based computation model which captures all known symbolic elimination algorithms in effective algebraic geometry and show the intrinsically exponential complexity character of elimination in this complexity model.
\end{abstract}

%\pagenumbering{arabic}

%%%%%%%%%%%%%%%%%%%%%%%%%%%%%%%%%%%%%%%%%%%%%%%%%%%%%%%%%%%%%%%%%%%%%%%%%%%%%
%----------------------------------------------------------------------------
\section{\large{Introduction}}
%----------------------------------------------------------------------------
%%%%%%%%%%%%%%%%%%%%%%%%%%%%%%%%%%%%%%%%%%%%%%%%%%%%%%%%%%%%%%%%%%%%%%%%%%%%%
%\setcounter{page}{1}

Modern elimination theory starts with Kronecker's 1882 paper \cite{Kro1882} where the argumentation is essentially constructive, i.e., algorithmic. Questions of efficiency of algorithms become only indirectly and marginally addressed in this paper. However, later criticism of Kronecker's approach to algebraic geometry emphasized the algorithmic inefficiency of his argumentation (\cite{M16}, \cite{vdW50}). In a series of more recent contributions, that started with \cite{Caniglia89} and ended up with \cite{Giusti1}, \cite{Giusti2}, \cite{HMW01} and \cite{GLS01}, it became apparent that this criticism is based on a too narrow interpretation of Kronecker's elimination method. In fact, these contributions are, implicitly or explicitly, based on this method, nonwithstanding that they also contain other views and ideas coming from commutative algebra and algebraic complexity theory. 

A turning point was achieved by the combination of a new, global view of Newton iteration with Kronecker's method (\cite{Giusti1}, \cite{Giusti2}). The outcome was that elimination polynomials, although hard to represent by their coefficients, allow a reasonably efficient encoding by evaluation algorithms. This circumstance suggests to represent in elimination algorithms polynomials not by their coefficients but by arithmetic circuits (see \cite{HeSie81}, \cite{Ka88} and \cite{FIK86} for the first steps in this direction). This idea became fully realized by the ``Kronecker'' algorithm for the resolution of polynomial equation systems over algebraically closed fields. The algorithm was anticipated in \cite{GHMP95}, \cite{Giusti1}, \cite{HMW01}, \cite{GLS01} and implemented in a software package of identical name (see \cite{Lecerf}).

The results of the present paper imply that the complexity of the Kronecker algorithm is asymptotically optimal under reasonable architectural assumptions. In particular, we exhibit in Section \ref{A hard elimination problem} of this paper an infinite family of arithmetic input circuits encoding efficiently certain elimination problems such that any algorithm, designed by commonly accepted rules of software engineering, requires exponential time to solve these problems. In order to establish such a result we need a computation model. In this contribution we put main emphasis on this task. 

In the same spirit we prove in Section \ref{independent of the model} that there exists an infinite family of parameter dependent elimination polynomials which, under a certain requirement of branching--freeness (called ``robustness''), need arithmetic circuits of exponential size for their evaluation, whereas the circuit size of the corresponding elimination problems grows only polynomially. 

In the past, many attempts to show the non--polynomial character of the elimination of just one existential quantifier block in the arithmetic circuit based elementary language over $\C$, employed the reduction to the claim that an appropriate candidate family of specific polynomials was hard to evaluate (this approach was introduced in \cite{HeintzMorgenstern93} and became adapted to the BSS model in \cite{SS95}). We give here the first example of such a family.

We finish this paper with a discussion of the concept of an approximative algorithm and establish a link between this notion and our computation model.

\section[Concepts and tools from algebraic geometry]{\large{Concepts and tools from algebraic geometry}}\label{sec: geometry}
%Capitulo tecnico sobre geometria algebraica

In this section, we use freely standard notions and notations from commutative algebra and algebraic geometry. These can be found for example in \cite{Lang93}, \cite{ZaSa60}, \cite{Kunz85} and \cite{Shafarevich94}. In Sections~\ref{sec: geometry-basic} and~\ref{sec: geometry-robust}, we introduce the notions and definitions which constitute our fundamental tool for the modelling of elimination problems and algorithms. Most of these notions and their definitions are taken from \cite{GHMS09}.

\subsection{Basic notions and notations} \label{sec: geometry-basic}

For any $n\in\N$, we denote by $\A^n:=\A^n(\C)$ the $n$--dimensional affine space $\C^n$ equipped with its respective Zariski and Euclidean topologies over $\C$. In algebraic geometry, the Euclidean topology of $\A^n$ is also called the {\em strong topology}. We shall use this terminology only exceptionally. 

Let $X_1,\ldots,X_n$ be indeterminates over $\C$ and let $X:=(X_1,\ldots,X_n)$. We denote by 
$\C[X]$ the ring of polynomials in the variables $X$ with complex coefficients.

Let $V$ be a closed affine subvariety of $\A^n$. As usual, we write $\dim V$ for the dimension of the variety $V$. 

%Let $C_1,\dots,C_s$ be the irreducible components of $V$. For $1\leq j\leq s$ we define the degree of $C_j$ as the number of points which arise when we intersect $C_j$ with $\dim C_j$ many generic affine hyperplanes of $\A^n$. Observe that this number is a well--determined positive integer which we denote by $\deg C_j$. The \emph{(geometric) degree} $\deg V$ of $V$ is defined by $\deg V:= \sum_{1\leq j\leq s} \deg C_j$. This notion of degree satisfies the so called B\'ezout Inequality. Namely, for another closed affine subvariety $W$ of $\A^n$ we have $\deg V\cap W \leq \deg V \cdot \deg W.$ 

%For details we refer to \cite{HeintzTesis83}, where the notion of geometric degree was introduced and the B\'ezout Inequality was proved for the first time (other references are \cite{Fu84} and \cite{Vo84}).    

For $f_1,\dots,f_s,g \in \C[X]$ we  shall use the notation $\{f_1=0,\dots,f_s=0\}$ 
in order to denote the closed affine subvariety $V$ of $\A^n$ defined by $f_1,\dots,f_s$ 
and the notation $\{f_1=0,\dots,f_s=0,g\not=0\}$ in order to denote the Zariski open subset $V_g$ of $V$ defined by the intersection of $V$ with the complement of $\{g=0\}$. %
Observe that $V_g$ is a locally closed affine subvariety of $\A^n$ whose coordinate ring
is the localization ${\C[V]}_g$ of $\C[V]$.

We denote by $I(V):= \{f \in \C[X]: f(x)=0$ for any $x \in V \}$ the ideal of definition of $V$ in $\C[X]$ and by $\C[V]:= \{\varphi:V \to \C \ws; \textrm{ there exists\ } f \in \C[X] \textrm{\ with\ } {\varphi}(x)=f(x) \textrm{\ for any\ } x \in V\}$ its coordinate ring. Observe that $\C[V]$ is isomorphic to the quotient $\C$--algebra $\C[V]:=\C[X]/I(V)$. If $V$ is irreducible, then $\C[V]$ is zero--divisor free and we denote by $\C(V)$ the field formed by the rational functions of $V$ with maximal domain ($\C(V)$ is called the rational function field of $V$). Observe that $\C(V)$ is isomorphic to the fraction field of the integral domain $\C[V]$.

In the general situation where $V$ is an arbitrary closed affine
subvariety of $\A^n$, the notion of a rational function of $V$ has
also a precise meaning. The only point to underline is that the domain, say $U$, of a rational function of $V$ has to be a maximal Zariski open and dense subset of $V$ to which the given rational function can be extended. In particular, $U$ has a nonempty intersection with any of the irreducible components of $V$.

As in the case where $V$ is irreducible, we denote by $\C(V)$ the $\C$--algebra formed by the rational functions of $V$. In algebraic terms, $\C(V)$ is the total quotient ring of $\C[V]$ and is isomorphic to the direct product of the rational function fields of the irreducible components of $V$.

Let be given a partial map $\phi: V \dashrightarrow W$, where $V$ and $W$ are closed subvarieties of some affine spaces $\A^n$ and $\A^m$, and let $\phi_1\klk\phi_m$ be the components of $\phi$. With these notations we have the following definitions:

\begin{definition}[Polynomial map] \label{def: pol. map}
The map $\phi$ is called a morphism of affine varieties or just polynomial map if the complex valued functions $\phi_1\klk\phi_m$ belong to $\C[V]$. Thus, in particular, $\phi$ is a total map.
\end{definition}
\mydefinitions{\label{def: pol. map} Polynomial map}
 
\begin{definition}[Rational map] \label{def: rational map}
We call $\phi$ a rational map of $V$ to $W$, if the domain $U$ of $\phi$ is a Zariski open and dense subset of $V$ and $\phi_1\klk\phi_m$ are the restrictions of suitable rational functions of $V$ to $U$. 
\end{definition}
\mydefinitions{\label{def: rational map} Rational map}

Observe that our definition of a rational map differs from the usual one in algebraic geometry, since we do not require that the domain $U$ of $\phi$ is maximal. Hence, in the case $m:=1$, our concepts of rational function and rational map do not coincide (see also \cite{GHMS09}).

%%%%%%%%%%%%%%%%%%%%%%%%%%%%%%%%%%%%%%%%%%%%%%%%%%%%%%%%%%%%%%%%%%%%%%%%%%%%%%%%%%%%%%%%%%%%
%%%%%%%%%%%%%%%%%%%%%%%%%%%%%%%%%%%%%%%%%%%%%%%%%%%%%%%%%%%%%%%%%%%%%%%%%%%%%%%%%%%%%%%%%%%%
\subsection{Constructible sets and constructible maps}
\label{sec: geometry-basic-constuctible}
Let $\mathcal{M}$ be a subset of some affine space $\A^n$ and, for a given nonnegative integer $m$, let $\phi:\mathcal{M}\dashrightarrow \A^m$ be a partial map. 

\begin{definition}[Constructible set] \label{def: constr. set}
We call the set $\mathcal{M}$ {\em constructible} if $\mathcal{M}$ is definable by a Boolean combination of polynomial equations.
\end{definition} 
\mydefinitions{\label{def: constr. set} Constructible set}

A basic fact we shall use in the sequel is that if
$\mathcal{M}$ is constructible, then its Zariski closure is equal to
its Euclidean closure (see, e.g., \cite{Mumford88}, Chapter I, \S 10, Corollary
1). In the same vein we have the following definition.
% or \cite[Lemma 12.5.3]{SoWa05}).

\begin{definition}[Constructible map] \label{def: constr. map}
We call the partial map $\phi$ {\em constructible} if the graph of $\phi$ is constructible as a subset of the affine space $\A^n\times \A^m$.
\end{definition} 
\mydefinitions{\label{def: constr. map} Constructible map}

We say that $\phi$ is {\em polynomial} if $\phi$ is the restriction of a morphism of affine varieties $\A^n\to \A^m$ to a constructible subset $\M$ of $\A^n$ and hence a total
map from $\M$ to $\A^m$. Furthermore, we call $\phi$ a {\em
rational} map of $\mathcal{M}$ if the domain $U$ of $\phi$ is
contained in $\mathcal{M}$ and $\phi$ is the restriction to
$\mathcal{M}$ of a rational map of the Zariski closure
$\ol{\mathcal{M}}$ of $\mathcal{M}$. In this case $U$ is a Zariski
open and dense subset of $\mathcal{M}$.

Since the elementary, i.e., first--order theory of algebraically closed fields with constants in $\C$ admits quantifier elimination, constructibility means just elementary definability. In particular, $\phi$ is constructible implies that the domain and the image of $\phi$ are constructible subsets of $\A^n$ and $\A^m$, respectively.

\begin{remark}
\label{remark constructible}
A partial map $\phi:\M\dashrightarrow\A^m$ is constructible if and only if it is piecewise rational. If $\phi$ is constructible there exists a Zariski open and dense subset $U$ of $\M$ such that the restriction $\phi|_U$ of $\phi$ to $U$ is a rational map of $\M$ (and of $\MM$).    
\end{remark}

For details we refer to \cite{GHMS09}, Lemma 1.

%Observe that $\M$ and $\A^m$ have two topologies, namely the Zariski topology and the Euclidean one. Witch respect to both topologies, the closure of $\M$ in $\A^n$ gives rise to the same subset of $\A^n$, which we denote by $\MM$ (see [Scha], Chapter VII, Section 2.1, Lemma 1).

\subsection{Geometrically robust constructible maps}
\label{sec: geometry-robust}
The main mathematical tool of this paper is the notion of geometrical robustness which we are going to introduce now.

Let $\M$ be a constructible subset of the affine space $\A^n$ and let $\phi:\M\to\A^m$ be a (total) constructible map with components $\phi_1,\dots,\phi_m$.

We consider now the Zariski closure $\ol{\mathcal M}$ of the constructible subset $\mathcal M$ of $\A^n$. Observe that $\ol{\mathcal M}$ is a closed affine subvariety of $\A ^n$ and that we may interpret $\C(\ol{\mathcal M})$ as a $\C[\ol{\mathcal{M}}]$--module (or $\C[\ol{\mathcal{M}}]$--algebra). 

Fix now an arbitrary point $x$ of $\ol{\mathcal M}$. By $\mathfrak{M}_x$ we denote the maximal ideal
of coordinate functions of $\C[\ol{\mathcal M}]$ which vanish at
the point $x$. By $\C[\ol{\mathcal M}]_{\mathfrak{M}_x}$ we denote the local
$\C$--algebra of the variety $\MM$ at the point $x$, i.e., the
localization of $\C[\MM]$ at the maximal ideal $\mathfrak{M}_x$. By $\C(\MM )_{\mathfrak{M}_x}$ we denote the localization of the $\C[\MM]$--module $\C(\MM )$ at $\mathfrak{M}_x$.

Following Remark \ref{remark constructible}, we may interpret $\phi_1,\ldots,\phi_m$ as rational functions of the
affine variety $\MM$ and therefore as elements of the total fraction ring $\C(\MM)$ of $\C[\MM]$. Thus $\C[\MM][\phi_1,\ldots,\phi_m]$ and
$\C[\MM]_{\mathfrak{M}_x}[\phi_1,\ldots,\phi_m]$ are
$\C$--subalgebras of $\C(\MM)$ and $\C(\MM)_{\mathfrak{M}_x}$ which contain $\C[\MM]$ and $\C[\MM]_{\mathfrak{M}_x}$, respectively.

The following result establishes for constructible maps a bridge between a topological and an algebraic notion. It will be fundamental in the context of this paper.

\begin{namedtheorem}[Theorem--Definition]
%\begin{theorem}
\label{theorem-definition}
Let notations and assumptions be as before. We call the constructible map $\phi:\M\to\A^m$ geometrically robust if $\phi$ is continuous with respect to the Euclidean topologies of $\M$ and $\A^m$ or equivalently, if $\phi_1,\dots,\phi_m$, interpreted as rational functions of the affine variety $\MM$, satisfy at any point $x\in\M$ the following two conditions: 
\begin{itemize}
	\item[(i)] $\C[\ol{\mathcal{M}}]_{\mathfrak{M}_x}[\phi_1,\dots,\phi_m]$ is a finite $\C[\ol{\mathcal{M}}]_{\mathfrak{M}_x}$--module.	 
	\item[(ii)] $\C[\ol{\mathcal{M}}]_{\mathfrak{M}_x}[\phi_1,\dots,\phi_m]$ is a local $\C[\ol{\mathcal{M}}]_{\mathfrak{M}_x}$--algebra whose maximal ideal is generated by $\mathfrak{M}_x$ and $\phi_1-\phi_1(x),\dots,\phi_m-\phi_m(x)$.
\end{itemize}
%\end{theorem}
\end{namedtheorem}

For a proof of this result, which is based on Zariski's Main Theorem (\cite{Iversen73}, \S IV.2) we refer to \cite{HKR11} (see also \cite{CaGiHeMaPa03} and \cite{GHMS09}).

From the topological definition of a geometrically robust constructible map one deduces immediately the following statement.

\begin{corollary}
\label{proposition 1}
If we restrict a geometrically robust constructible map to a constructible subset of its domain of definition we obtain again a geometrically robust map. Moreover the composition and the cartesian product of two geometrically robust constructible maps are geometrically robust. The geometrically robust constructible functions form a commutative $\C$--algebra which contains the polynomial functions.
\end{corollary}

The origin of the concept of a geometrically robust map can be found, implicitly, in \cite{GH01}. It was introduced explicitly for constructible maps with irreducible domains of definition in \cite{GHMS09}, where it is used to analyze the complexity character of multivariate Hermite--Lagrange interpolation.

%%Hasta aqui pagina 6

\section{Robust parameterized arithmetic circuits}
We shall use freely standard concepts from algebraic complexity theory which can be found in $\cite{Burgisser97}$.

Let us fix natural numbers $n$ and $r$, indeterminates $X_1,\dots ,X_n$ and a non--empty constructible subset $\mathcal{M}$ of $\A^r$. By $\pi_1,\dots ,\pi_r$ we denote the restrictions to $\mathcal{M}$ of the canonical projections $\A^r\to\A^1$.

A \emph{(by $\mathcal{M}$) parameterized arithmetic circuit $\beta$} (with \emph{basic parameters $\pi_1,\dots ,\pi_r$} and \emph{inputs $X_1,\dots ,X_n$}) is a labelled directed acyclic graph (labelled DAG) satisfying the following conditions:\\
each node of indegree zero is labelled by a scalar from $\C$, a basic parameter $\pi_1,\dots ,\pi_r$ or a input variable $X_1,\dots ,X_n$. Following the case, we shall refer to the \emph{scalar, basic parameter} and (standard) \emph{input} nodes of $\beta$. All other nodes of $\beta$ have indegree two and are called \emph{internal}. They are labelled by arithmetic operations (addition, subtraction, multiplication, division). A \emph{parameter} node of $\beta$ depends only on scalar and basic parameter nodes, but not on any input node of $\beta$ (here ``dependence'' refers to the existence of a connecting path). An addition or multiplication node whose two ingoing edges depend on an input is called \emph{essential}. The same terminology is applied to division nodes whose second argument depends on an input. Moreover, at least one circuit node becomes labelled as output. Without loss of generality we may suppose that all nodes of outdegree zero are outputs of $\beta$.	

We consider $\beta$ as a syntactical object which we wish to equip with a certain semantics. In principle there exists a canonical evaluation procedure of $\beta$ assigning to each node a rational function of $\mathcal{M}\times \A^n$ which, in case of a parameter node, may also be interpreted as a rational function of $\mathcal{M}$. In either situation we call such a rational function an \emph{intermediate result} of $\beta$. 

The evaluation procedure may fail if we divide at some node an intermediate result by another one which vanishes on a Zariski dense subset of a whole irreducible component of $\mathcal{M}\times \A^n$. If this occurs, we call the labelled DAG $\beta$ \emph{inconsistent}, otherwise \emph{consistent}. 

If nothing else is said, we shall from now on assume that $\beta$ is a consistent parameterized arithmetic circuit. The intermediate results associated with output nodes will be called \emph{final results} of $\beta$. 

We call an intermediate result associated with a parameter node a \emph{parameter} of $\beta$ and interpret it generally as a rational function of $\mathcal{M}$. A parameter associated with a node which has an outgoing edge into a node which depends on some input of $\beta$ is called \emph{essential}. In the sequel we shall refer to the constructible set $\mathcal{M}$ as the \emph{parameter domain} of $\beta$.

We consider $\beta$ as a syntactic object which represents the final results of $\beta$, i.e., the rational functions of $\mathcal{M}\times\A^n$ assigned to its output nodes. 

%In this way becomes introduced an abstraction function which associates $\beta$ with these rational functions. This abstraction function assigns therefore to $\beta$ a rational map $\mathcal{M}\times\A^n \dashrightarrow \A^q$, where $q$ is the number of output nodes of $\beta$. On its turn, this rational map may also be understood as a (by $\mathcal{M}$) parameterized family of rational maps $\A^n \dashrightarrow \A^q$.

Now we suppose that the parameterized arithmetic circuit $\beta$ has been equipped with an additional structure, linked to the semantics of $\beta$. We assume that for each node $\rho$ of $\beta$ there is given a \emph{total} constructible map $\mathcal{M}\times \A^n \to \A^1$ which extends the intermediate result associated with $\rho$. In this way, if $\beta$ has $K$ nodes, we obtain a total constructible map $\Omega:\mathcal{M}\times \A^n \to \A^K$ which extends the rational map $\mathcal{M}\times \A^n \dashrightarrow \A^K$ given by the labels at the indegree zero nodes and the intermediate results of $\beta$.
    
\begin{definition}[Robust circuit] \label{def: robust circuit}
Let notations and assumptions be as before. The pair $(\beta,\Omega)$ is called a robust parameterized arithmetic circuit if the constructible map $\Omega$ is geometrically robust. 
\end{definition}
\mydefinitions{\label{def: robust circuit} Robust circuit}

Observe that the above rational map $\M\times\A^n \dashrightarrow \A^K$ can be extended \emph{to at most one} geometrically robust constructible map $\Omega:\M\times\A^n\to\A^K$. Therefore we shall apply from now on the term ``robust'' also to the circuit $\beta$.

Robust parameterized arithmetic circuits may be restricted as follows:\\
Let $\mathcal{N}$ be a constructible subset of $\mathcal{M}$ and suppose that $(\beta,\Omega)$ is robust. Then Corollary \ref{proposition 1} implies that the restriction $\Omega\vert_{\mathcal{N}\times\A^n}$ of the constructible map $\Omega$ to $\mathcal{N}\times\A^n$ is still a geometrically robust constructible map.

This implies that $(\beta,\Omega)$ induces a by $\mathcal{N}$ parameterized arithmetical circuit $\beta_{\mathcal{N}}$. Observe that $\beta_{\mathcal{N}}$ may become inconsistent. If $\beta_{\mathcal{N}}$ is consistent then $(\beta_{\mathcal{N}}, \Omega\vert_{\mathcal{N}\times\A^n})$ is robust. The nodes where the evaluation of $\beta_{\mathcal{N}}$ fails correspond to divisions of zero by zero which may be replaced by so called approximative algorithms having unique limits (see \cite{HKR11}, Section 3.3.2 and Section \ref{A hard elimination problem} below). These limits are given by the map $\Omega\vert_{\mathcal{N}\times\A^n}$. We call $(\beta_{\mathcal{N}}, \Omega\vert_{\mathcal{N}\times\A^n})$, or simply $\beta_{\mathcal{N}}$, the \emph{restriction} of $(\beta,\Omega)$ or $\beta$ to $\mathcal{N}$. 

We say that the parameterized arithmetic circuit $\beta$ is \emph{totally division--free} if any division node of $\beta$ corresponds to a division by a non--zero complex scalar.

We call $\beta$ \emph{essentially division--free} if only parameter nodes are labelled by divisions. Thus the property of $\beta$ being totally division--free implies that $\beta$ is essentially division--free, but not vice versa. Moreover, if $\beta$ is totally division-free, the rational map given by the intermediate results of $\beta$ is polynomial and therefore a geometrically robust constructible map. Thus, any by $\mathcal{M}$ parameterized, totally division--free circuit is in a natural way robust.

In the sequel, we shall need the following elementary fact. 

\begin{lemma}
\label{lemma intermediate results}
Let notations and assumptions be as before and suppose that the parameterized arithmetic circuit $\beta$ is robust. Then all intermediate results of $\beta$ are polynomials in $X_1,\dots,X_n$ over the $\C$--algebra of geometrically robust constructible functions defined on $\mathcal{M}$.
\end{lemma}

For a proof of Lemma \ref{lemma intermediate results} we refer to \cite{HKR11}, Section 3.1.

The statement of this lemma should not lead to confusions with the notion of an essentially division--free parameterized circuit. We say just that the intermediate results of $\beta$ are polynomials in $X_1,\dots , X_n$ and do not restrict the type of arithmetic operations contained in $\beta$ (as we did defining the notion of an essentially division--free parameterized circuit). 

To our parameterized arithmetic circuit $\beta$ we may associate different complexity measures and models. In this paper we shall mainly be concerned with \emph{sequential computing time}, measured by the \emph{size} of $\beta$. Here we refer with ``size'' to the number of internal nodes of $\beta$ which count for the given complexity measure. Our basic complexity measure is the \emph{non--scalar} one (also called \emph{Ostrowski measure}) over the ground field $\C$. This means that we count, at unit costs, only essential multiplications and divisions (involving basic parameters or input variables in both arguments in the case of a multiplication and in the second argument in the case of a division), whereas $\C$--linear operations are free (see \cite{Burgisser97} for details). 

\subsection{Operations with robust parameterized arithmetic circuits}

\subsubsection{The operation join}

Let $\gamma_1$ and $\gamma_2$ be two robust parameterized arithmetic circuits with parameter domain $\M$ and suppose that there is given a one--to--one correspondence $\lambda$ which identifies the output nodes of $\gamma_1$ with the input nodes of $\gamma_2$ (thus they must have the same number). Using this identification we may now join the circuit $\gamma_1$ with the circuit $\gamma_2$ in order to obtain a new parameterized arithmetic circuit $\gamma_2*_{\lambda}\gamma_1$ with parameter domain $\M$. The circuit $\gamma_2*_{\lambda}\gamma_1$ has the same input nodes as $\gamma_1$ and the same output nodes as $\gamma_2$ and one deduces easily from Lemma \ref{lemma intermediate results} and Corollary \ref{proposition 1} that the circuit $\gamma_2*_{\lambda}\gamma_1$ is robust and represents a composition of the rational maps defined by $\gamma_1$ and $\gamma_2$, if $\gamma_2*_{\lambda}\gamma_1$ is consistent. The (consistent) circuit $\gamma_2*_{\lambda}\gamma_1$ is called the (consistent) \emph{join} of $\gamma_1$ with $\gamma_2$.

Observe that the final results of a given robust parameterized arithmetic circuit may constitute a vector of parameters. The join of such a circuit with another robust parameterized arithmetic circuit, say $\beta$, is again a robust parameterized arithmetic circuit which is called an \emph{evaluation} of $\beta$. Hence, mutatis mutandis, the notion of join of two routines includes also the case of circuit evaluation.

\subsubsection{The operations reduction and broadcasting}

We describe now how, based on its semantics, a given parameterized arithmetic circuit $\beta$ with parameter domain $\M$ may be rewritten as a new circuit over $\M$ which computes the same final results as $\beta$.

The resulting two rewriting procedures, called \emph{reduction} and \emph{broadcasting}, will neither be unique, nor generally confluent. To help understanding, the reader may suppose that there is given an (efficient) algorithm which allows identity checking between intermediate results of $\beta$. However, we shall not make explicit reference to this assumption. We are now going to explain the first rewriting procedure.

Suppose that the parameterized arithmetic circuit $\beta$ computes at two different nodes, say $\rho$ and $\rho'$, the same intermediate result. Assume first that $\rho$ neither depends on $\rho'$, nor $\rho'$ on $\rho$. Then we may erase $\rho'$ and its two ingoing edges (if $\rho'$ is an internal node) and draw an outgoing edge from $\rho$ to any other node of $\beta$ which is reached by an outgoing edge of $\rho'$. If $\rho'$ is an output node, we label $\rho$ also as output node. Observe that in this manner a possible indexing of the output nodes of $\beta$ may become changed but not the final results of $\beta$ themselves.

Suppose now that $\rho'$ depends on $\rho$. Since the DAG $\beta$ is acyclic, $\rho$ does not depend on $\rho'$. We may now proceed in the same way as before, erasing the node $\rho'$.

Let $\beta'$ be the parameterized arithmetic circuit obtained, as described before, by erasing the node $\rho'$. Then we call $\beta'$ a \emph{reduction} of $\beta$ and call the way we obtained $\beta'$ from $\beta$ a \emph{reduction step}. A \emph{reduction procedure} is a sequence of successive reduction steps.

One sees now easily that a reduction procedure applied to $\beta$ produces a new parameterized arithmetic circuit $\beta^*$ (also called a \emph{reduction} of $\beta$) with the same basic parameter and input nodes, which computes the same final results as $\beta$ (although their possible indexing may be changed). Moreover, if $\beta$ is a robust parameterized circuit, then $\beta^*$ is robust too. Observe also that in the case of robust parameterized circuits our reduction commutes with restriction. 

Now we introduce the second rewriting procedure.

Let assumptions and notations be as before and let be given a set $P$ of nodes of $\beta$ and a robust parameterized arithmetic circuit $\gamma$ with parameter domain $\mathcal{M}$ and $\# P$ input nodes, namely for each $\rho\in P$ one which becomes labelled by a new input variable $Y_{\rho}$. We obtain a new parameterized arithmetic circuit, denoted by $\gamma*_P\beta$, when we join $\gamma$ with $\beta$, replacing for each $\rho\in P$ the input node of $\gamma$, which is labelled by the variable $Y_{\rho}$, by the node $\rho$ of $\beta$. The output nodes of $\beta$ constitute also the output nodes of $\gamma*_P\beta$. Thus $\beta$ and $\gamma*_P\beta$ compute the same final results. Observe that $\gamma*_P\beta$ is robust if it is consistent. We call the circuit $\gamma*_P\beta$ and all its reductions \emph{broadcastings} of $\beta$. Thus broadcasting a robust parameterized arithmetic circuit means rewriting it using only valid polynomial identities.

If we consider arithmetic circuits as computer programs, then reduction and broadcasting represent a kind of program transformations.

%hasta aqui pagina 12

%------------------------------------------------------------------------------------------
%\section[Applications of the extended computation model to complexity issues of effective elimination theory]{\large{Applications of the extended computation model to\\ complexity issues of effective elimination theory}}\label{sec:Model-Applications}

\section{A family of hard elimination polynomials} 
\label{independent of the model}

As a major result of this paper we are now going to exhibit an infinite family of parameter dependent elimination polynomials which require essentially division--free, \emph{robust} parameterized arithmetic circuits of exponential size for their evaluation, whereas the circuit size of the corresponding input problems grows only polynomially. The proof of this result, which is absolutely new in his kind, is astonishly elementary and simple.

Let $T,U_1,\dots,U_n$ and $X_1,\dots,X_n$ be indeterminates and let $U:=(U_1,\dots,U_n)$ and $X:=(X_1,\dots,X_n)$. Consider for given $n\in\N$ the polynomial $H^{(n)}:= \sum_{1 \leq i \leq n} 2^{i-1}$ $X_i + T \prod_{1 \leq i \leq n} (1+ (U_i-1)X_i)$. Observe that $H^{(n)}$ can be evaluated using $n-1$ non--scalar multiplications involving $X_1,\dots,X_n$.  

The set $\mathcal{O}:=\{ \sum_{1 \leq i \leq n} 2^{i-1}X_i + t \prod_{1 \leq i \leq n} (1+ (u_i-1)X_i);(t,u_1,\dots,u_n)\in\A^{n+1} \}$ is contained in a finite--dimensional $\C$--linear subspace of $\C[X]$ and therefore $\mathcal{O}$ and its closure $\ol{\mathcal{O}}$ are constructible sets.

From \cite{GHMS09}, Section 3.3.3 we deduce the following facts:\\
there exist $K:=16n^2+2$ integer points $\xi_1,\dots,\xi_K \in\Z^n$ of bit length at most $4n$ such that for any two polynomials $f,g\in \ol{\mathcal{O}}$ the equalities $f(\xi_k)=g(\xi_k), 1\leq k \leq K$, imply $f=g$. Thus the polynomial map $ \Xi:\ol{\mathcal{O}}\to\A^K$ defined for $f\in \ol{\mathcal{O}}$ by $\Xi(f):= (f(\xi_1),\dots,f(\xi_K))$ is injective. Moreover $\mathcal{M}:=\Xi(\mathcal{O})$ is an irreducible constructible subset of $\A^K$ and we have $\ol{\mathcal{M}}=\Xi(\ol{\mathcal{O}})$. Finally, the constructible map $\phi:=\Xi^{-1}$, which maps $\mathcal{M}$ onto $\mathcal{O}$ and $\ol{\mathcal{M}}$ onto $\ol{\mathcal{O}}$, is a restriction of a geometrically robust map and therefore by Corollary \ref{proposition 1} itself geometrically robust.

For $\epsilon\in\{ 0,1 \}^n$ we denote by $\phi_\epsilon$ the map $\ol{\mathcal{M}}\to\A^1$ which assigns to each point $v\in\ol{\mathcal{M}}$ the value $\phi(v)(\epsilon)$. From Corollary \ref{proposition 1} we conclude that $\phi_{\epsilon}$ is a geometrically robust constructible function which belongs to the function field $\C(\ol{\mathcal{M}})$ of the irreducible algebraic variety $\ol{\mathcal{M}}$.

Observe that for $t\in\A^1$ and $u\in\A^n$ the identities $\phi_{\epsilon}(\Xi(H^{(n)}(t,u,X)))=\phi(\Xi(H^{(n)}(t,u,X)))(\epsilon)= ((\Xi^{-1}\circ\Xi)(H^{(n)}(t,u,X)))(\epsilon)= H^{(n)}(t,u,\epsilon)$ hold.

Let $P^{(n)}:= \prod_{\epsilon\in \{ 0,1\}^n } (Y-\phi_\epsilon)$. Then $P^{(n)}$ is a geometrically robust constructible function which maps $\ol{\mathcal{M}}\times\A^1$ (and hence $\mathcal{M}\times\A^1$) into $\A^1$. 

Consider now the polynomial $F^{(n)}:= \prod_{\epsilon\in \{ 0,1\}^n } (Y- H^{(n)}(T,U,\epsilon)) = 
\prod_{0 \leq j \leq 2^n-1}(Y-(j + T\prod_{1 \leq i\leq n} U_i^{[j]_i}))$, where $[j]_i$ denotes the $i$--th digit of the binary representation of the integer $j$, $0 \leq j \leq 2^n-1$, $1 \leq i \leq n$. We have for $t\in\A^1$ and $u\in\A^n$ the identities 

\begin{equation}
\begin{array}{l}
\displaystyle P^{(n)}(\Xi(H^{(n)}(t,u,X)),Y) = \prod_{\epsilon\in \{ 0,1 \}^n} (Y - \phi_\epsilon(\Xi(H^{(n)}(t,u,X))))=\\
\displaystyle  \prod_{\epsilon\in \{ 0,1 \}^n} (Y-H^{(n)}(t,u,\epsilon)) = F^{(n)}(t,u,Y)
\end{array} 
\label{equ (*)}
\end{equation}

Let $S_1,\dots,S_K$ be new indeterminates and observe that the existential first order formula of the elementary theory of $\C$, namely
\begin{equation}
\begin{array}{l}
\displaystyle (\exists X_1)\dots(\exists X_n)(\exists T)(\exists U_1)\dots(\exists U_n) (X_1^2-X_1=0 \wedge\dots \wedge X_n^2-X_n=0 \wedge\\
\displaystyle  \bigwedge_{1 \leq j \leq K} S_j=H^{(n)}(T,U,\xi_j) \wedge Y=H^{(n)}(T,U,X))
\end{array} 
\label{equ (**)}
\end{equation}

describes the constructible subset $\{ (s,y)\in\A^{K+1}; s\in\mathcal{M},y\in\A^1,P^{(n)}(s,y)=0 \}$ of $\A^{K+1}$. Moreover, $P^{(n)}$ is the greatest common divisor in $\C(\ol{\mathcal{M}})[Y]$ of all polynomials of $\C[\ol{\mathcal{M}}][Y]$ which vanish identically on the constructible subset of $\A^{K+1}$ defined by the formula (\ref{equ (**)}). Hence $P^{(n)}\in\C(\ol{\mathcal{M}})[Y]$ is a (parameterized) \emph{elimination polynomial}.

Observe that the polynomials contained in the formula (\ref{equ (**)}) can be represented by a totally division--free arithmetic circuit of size $O(n^3)$. Therefore, the formula (\ref{equ (**)}) is also of size $O(n^3)$. 

\begin{theorem}
\label{theorem model independent}
Let notations and assumptions be as before and let $\gamma$ be an essentially division--free, robust parameterized arithmetic circuit with domain of definition $\mathcal{M}$ such that $\gamma$ evaluates the elimination polynomial $P^{(n)}$.

Then $\gamma$ performs at least $\Omega(2^{\frac{n}{2}})$ essential multiplications and at least $\Omega(2^n)$ multiplications with parameters.
\end{theorem}

\begin{proof}
We fix the natural number $n$. Let us write $H:=H^{(n)}=\sum_{1 \leq i\leq n} 2^{i-1}X_i + \prod_{1 \leq i\leq n} T(U_i-1)X_i $ as a polynomial in the main indeterminates $X_1,\dots,X_n$ with coefficients $\theta_{\kappa_1,\dots,\kappa_n}\in\C[T,U]$, $\kappa_1,\dots,\kappa_n\in\left\{0,1\right\}$, namely 
$$H=\sum_{\kappa_1,\dots,\kappa_n\in\left\{0,1\right\}} \theta_{\kappa_1,\dots,\kappa_n} X_1^{\kappa_1},\dots,X_n^{\kappa_n}.$$
Observe that for $\kappa_1,\dots,\kappa_n\in\left\{0,1\right\}$ the polynomial $\theta_{\kappa_1,\dots,\kappa_n}(0,U)\in \C[U]$ is of degree at most zero, i.e., a constant complex number, independent of $U_1,\dots,U_n$.

Let $\theta:=(\theta_{\kappa_1,\dots,\kappa_n})_{\kappa_1,\dots,\kappa_n\in\left\{0,1\right\}}$ and observe that the vector $\theta(0,U)$ is a fixed point of the affine space $\A^{2^n}$. We denote by $\mathfrak{M}$ the vanishing ideal of the $\C$--algebra $\C[\theta]$ at this point. We interpret $\theta$ as a geometrically robust constructible map $\A^{n+1}\to\A^{2^n}$ with (constructible) image $\mathcal{T}$.

Let us write $F:=F^{(n)}= \prod_{0\leq j\leq 2^{n}-1} (Y-(j+T \prod_{1\leq i\leq n}U_i^{\left[j\right]_i}))$ as a polynomial in the main indeterminate $Y$ with coefficients $\varphi_\kappa\in\C[T,U]$, $1\leq \kappa\leq 2^n$, namely
$F= Y^{2^n} + \varphi_1 Y^{2^n -1} +\dots+ \varphi_{2^n }.$ 

Observe that for $1\leq \kappa\leq 2^n $ the polynomial $\varphi_\kappa(0,U)\in\C[U]$ is of degree at most zero. Let $\lambda_\kappa:=\varphi_\kappa(0,U)$, $\lambda:=(\lambda_\kappa)_{1\leq \kappa\leq 2^n}$ and $\varphi:=(\varphi_\kappa)_{1\leq \kappa\leq 2^n }$. Observe that $\lambda$ is a fixed point of the affine space $\A^{2^n}$.

Let $\nu:\A^{n+1}\to\A^K$ be the polynomial map defined for $t\in\A^1$ and $u\in\A^n$ by $\nu(t,u):=\Xi(H(t,u,X))=(H(t,u,\xi_1),\dots,H(t,u,\xi_K))$. Observe that there exists a geometrically robust constructible map $\sigma:\mathcal{T}\to\A^K$ such that $\sigma\circ\theta=\nu$ holds. Since by assumption the parameterized arithmetic circuit $\gamma$ is essentially division--free and robust, there exists a geometrically robust constructible map $\psi$ defined on $\M$ such that the entries of $\psi$ constitute the essential parameters of the circuit $\gamma$. Moreover, for $m$ being the number of components of $\psi$, there exists a vector $\omega$ of $m$--variate polynomials over $\C$ such that the entries of $\omega(\psi)=\omega\circ\psi$ become the coefficients of the elimination polynomial $P:=P^{(n)}=\prod_{\epsilon\in\{ 0,1 \}^n}(Y-\phi_{\epsilon})$.

One sees easily that there exists a totally  division--free ordinary arithmetic circuit $\gamma '$ which evaluates the polynomials $H(T,U,\xi_1),\dots,H(T,U,\xi_K)$.

The join $\gamma*\gamma '$ of $\gamma '$ with $\gamma$ at the basic parameter nodes of $\gamma$ is an essentially division--free robust parameterized circuit with domain of definition $\A^{n+1}$ which by (\ref{equ (*)}) evaluates the polynomial $F(T,U,Y):=F^{(n)}(T,U,Y)$. The entries of the vector $\tilde{\nu}:=\psi\circ\nu$ constitute the essential parameters of the circuit $\gamma*\gamma'$ and the entries of $\omega\circ\tilde{\nu}=\omega\circ\psi\circ\nu$ become by  (\ref{equ (*)}) the coefficients of the polynomial $F(T,U,Y)$ with respect to $Y$. So we have $\varphi=\omega\circ\tilde{\nu}$.

Taking into account $\tilde{\nu}=\psi\circ\nu=\psi\circ\sigma\circ\theta$, Theorem--Definition \ref{theorem-definition} $(i)$ and \cite{GHMS09}, Corollary 12 we conclude that the entries of $\tilde{\nu}$ are polynomials of $\C[T,U]$ which are integral over the local $\C$--subalgebra $\C[\theta]_{\mathfrak{M}}$ of $\C(T,U)$.

Let $\mu\in\C[T,U]$ be such an entry. Then there exists an integer $s$ and polynomials $a_0,a_1,\dots,a_s \in\C[\theta]$ with $a_0\notin\mathfrak{M}$ such that the algebraic dependence relation
\begin{equation}
\label{(*)}
a_0\mu^s + a_1\mu^{s-1}+\dots+ a_s =0	
\end{equation}
is satisfied in $\C[T,U]$. From \eqref{(*)} we deduce the algebraic dependence relation
\begin{equation}
\label{(**)}
a_0(0,U)\mu(0,U)^s +a_1(0,U)\mu(0,U)^{s-1} +\dots+a_s(0,U) =0	
\end{equation}
in $\C[U]$.

Since the polynomials $a_0,a_1,\dots,a_s$ belong to $\C[\theta]$ and $\theta(0,U)$ is a fixed point of $\A^{2^n}$, we conclude that $\alpha_0:=a_0(0,U),\alpha_1 := a_1(0,U), \dots , \alpha_s:= a_s(0,U)$ are complex numbers. Moreover, $a_0\notin\mathfrak{M}$ implies $\alpha_0\neq 0$. 

Thus \eqref{(**)} may be rewritten into the algebraic dependence relation
\begin{equation}
\label{(***)}
\alpha_0\mu(0,U)^s +\alpha_1 \mu(0,U)^{s-1} +\dots+ \alpha_s =0
\end{equation}
in $\C[U]$ with $\alpha_0\neq 0$.

This implies that the polynomial $\mu(0,U)$ of $\C[U]$ is of degree at most zero. Therefore $w:=\tilde{\nu}(0,U)$ is a fixed point of the affine space $\A^m$. 

Recall that $\lambda=(\lambda_\kappa)_{1\leq \kappa\leq 2^n}$ with $\lambda_\kappa:=\varphi_\kappa(0,U)$, $1\leq \kappa\leq 2^n$, is a fixed point of the affine space $\A^{2^n}$.

%For $1\leq \kappa\leq 2^n$ we may write the polynomial $\varphi_\kappa\in\C[T,U]$ as follows:
%\begin{equation}
%\label{(****)}
%\varphi_\kappa=\lambda_\kappa + \Delta_\kappa T + \text{\ terms\ of\ higher\ degree\ in\ } T	
%\end{equation}

%with $\Delta_\kappa\in\C[U]$. 

From \cite{CaGiHeMaPa03}, Lemma 6 we deduce that for $1 \leq \kappa \leq 2^n$ the coefficient $\varphi_\kappa$ of $F$ is an element of $\C[T,U]$ of the form
\begin{equation}
\label{(****)}
\varphi_\kappa=\lambda_\kappa + T L_\kappa +\text{ terms of higher degree in $T$}
\end{equation}

%
%$$B_l = (-1)^l \sum_{l\leq j_1<\dots <j_l< 2^n} j_1\dots j_l +TL_l+\text{\ terms\ of\ higher\ degree\ in\ }T,$$

where $L_1,\dots,L_{2^n} \in \C[U]$ are $\C$--linearly independent.

%Choose now different complex numbers $\eta_1,\dots,\eta_{2^n}$ from $\C$ and observe that for $1\leq \kappa'\leq 2^n$ the identities
%
%$$\frac{\partial F}{\partial T} (0,U,\eta_{\kappa '}) =
%\sum_{1 \leq l\leq 2^n} L_l \eta_{\kappa '}^{2^n-l}$$
%and 
%$$\frac{\partial F}{\partial T} (0,U,\eta_{\kappa '}) =  \sum_{1\leq \kappa \leq 2^n} \Delta_\kappa \eta_{\kappa '}^{2^n - \kappa}$$
%hold.

%Since $L_1,\dots,L_{2^n}$ are $\C$--linearly independent, we deduce from the non--singularity of the Vandermonde matrix $(\eta_{\kappa '}^{2^n - l})_{1\leq l,\kappa'\leq 2^n}$ that the polynomials $\Delta_1,\dots,\allowbreak \Delta_{2^n}$ of $\C[U]$ are $\C$--linearly independent.

Consider now an arbitrary point $u\in\A^n$ and let $\epsilon_u:\A^1\to\A^m$ and $\delta_u:\A^1\to\A^{2^n}$ be the polynomial maps defined for $t\in\A^1$ by $\epsilon_u(t):=\tilde{\nu}(t,u)$ and $\delta_u(t):=\varphi(t,u)$. Then we have $\epsilon_u(0)=\tilde{\nu}(0,u)=w$ and $\delta_u(0)=\varphi(0,u)=\lambda$, independently of $u$. Moreover, from $\varphi=\omega\circ\tilde{\nu}$ we deduce $\delta_u=\omega\circ\epsilon_u$.%80

%Consider now an arbitrary point $u\in\A^n$ and let $\epsilon_u:\A^1\to\A^m$ and $\delta_u:\A^1\to\A^{2^n q}$ be the polynomial maps defined for $t\in\A^1$ by $\epsilon_u(t):=\nu^*(t,u)$ and $\delta_u(t):=\varphi(t,u)$. Then we have $\epsilon_u(0)=\nu^*(0,u)=w$ and $\delta_u(0)=\varphi(0,u)=\lambda$, independently of $u$. Moreover, from $\varphi=\omega\circ\nu^*$ we deduce $\delta_u=\omega\circ\epsilon_u$.%80

Thus \eqref{(****)} implies
\begin{equation}
\label{(+)}
(L_1(u),\dots,L_{2^n }(u))=\frac{\partial\varphi}{\partial t}(0,u)=\delta_u'(0)=(D\omega)_w(\epsilon_u'(0)),	
\end{equation}

where $(D \omega)_w$ denotes the (first) derivative of the $m$--variate polynomial map $\omega$ at the point $w\in\A^m$ and $\delta_u'(0)$ and $\epsilon_u'(0)$ are the derivatives of the parameterized curves $\delta_u$ and $\epsilon_u$ at the point $0\in\A^1$. We rewrite now \eqref{(+)} in matrix form, replacing $(D\omega)_w$ by the corresponding transposed Jacobi matrix $M\in\A^{m\times 2^n }$ and $\delta_u'(0)$ and $\epsilon_u'(0)$ by the corresponding points of $\A^{2^n }$ and $\A^m$, respectively. 

%where $(D \omega)_w$ denotes the (first) derivative of the $m$--variate polynomial map $\omega$ at the point $w\in\A^m$ and $\delta_u'(0)$ and $\epsilon_u'(0)$ are the derivatives of the parameterized curves $\delta_u$ and $\epsilon_u$ at the point $0\in\A^1$. We rewrite now \eqref{(+)} in matrix form, replacing $(D\omega)_w$ by the corresponding transposed Jacobi matrix $M\in\A^{m\times 2^n q}$ and $\delta_u'(0)$ and $\epsilon_u'(0)$ by the corresponding points of $\A^{2^n q}$ and $\A^m$, respectively. 

Then \eqref{(+)} takes the form
\begin{equation}
\label{(++)}
(L_1(u),\dots,L_{2^n }(u))= \epsilon_u'(0)M,	
\end{equation}
where the complex $(m\times 2^n )$--matrix $M$ is independent of $u$.

Since the polynomials $L_1,\dots,L_{2^n}\in\C[U]$ are $\C$--linearly independent, we may choose points $u_1,\dots,u_{2^n}\in\A^n$ such that the complex $(2^n \times 2^n)$--matrix  
$$N:=(L_{\kappa}(u_l))_{1\leq l,\kappa \leq 2^n}$$ 
has rank $2^n$.

Let $K$ be the complex $(2^n\times m)$--matrix whose rows are $\epsilon_{u_1}'(0),\dots,\epsilon_{u_{2^n}}'(0)$. 

Then \eqref{(++)} implies the matrix identity 
$$N=K\cdot M.$$
Since $N$ has rank $2^n$, the rank of the complex $(m\times 2^n)$--matrix $M$ is at least $2^n$. This implies
\begin{equation}
\label{(+++)}
m\geq 2^n.	
\end{equation}
Therefore the circuit $\gamma$ contains $m\geq 2^n$ essential parameters.

Let $L$ be the number of essential multiplications executed by the parameterized arithmetic circuit $\gamma$ and let $L'$ be the total number of multiplications of $\gamma$, excepting those by scalars from $\C$. Then, after a well--known standard rearrangement \cite{PS73} of $\gamma$, we may suppose without loss of generality, that there exists a constant $c>0$ (independent of the input circuit $\gamma$) such that $L\geq cm^{\frac{1}{2}}$ and $L'\geq cm$ holds.

Let $L$ be the number of essential multiplications executed by the parameterized arithmetic circuit $\gamma$ and let $L'$ be the total number of multiplications of $\gamma$, excepting those by scalars from $\C$. Then, after a well--known standard rearrangement \cite{PS73} of $\gamma$, we may suppose without loss of generality, that there exists a constant $c>0$ (independent of the input circuit $\gamma$ and the procedure $\mathcal{A}$) such that $L\geq cm^{\frac{1}{2}}$ and $L'\geq cm$ holds.

From the estimation \eqref{(+++)} we deduce now that the circuit $\gamma$ performs at least $\Omega(2^{\frac{n}{2}})$ essential multiplications and at least $\Omega(2^n)$ multiplications, including also multiplications with parameters. This finishes the proof of the theorem.
\end{proof}

\enter

Theorem \ref{theorem model independent} is essentially contained in the arguments of the proof of \cite{GH01}, Theorem 5 and \cite{CaGiHeMaPa03}, Theorem 4.

Observe that a quantifier--free description of $\mathcal{M}$ by means of circuit represented polynomials, together with an essentially division--free, robust parameterized arithmetic circuit $\gamma$ with domain of definition $\mathcal{M}$, which evaluates the elimination polynomial $P^{(n)}$ captures the intuitive meaning of an algorithmic solution of the elimination problem described by the formula (\ref{equ (**)}), when we restrict our attention to solutions of this kind and minimize the number of equations and branchings. In particular the circuit $\gamma$ can be evaluated for any input point $(s,y)$ with $s\in\mathcal{M}$ and $y\in\C$ and the intermediate results of $\gamma$ are polynomials of $\C(\ol{\mathcal{M}})[Y]$ whose coefficients are geometrically robust constructible functions defined on $\mathcal{M}$.

%doc1

With respect to the indeterminate $Y$, the coefficients of the polynomial $P^{(n)}\in\C(\ol{\M})[Y]$ are geometrically robust constructible functions of the parameter domain $\M$. In order to consider $P^{(n)}$ as an elimination polynomial as we did, the reader might expect that the coefficients of $P^{(n)}$ should belong, for any point $s\in\M$, to the local ring of $\ol{\M}$ at $s$. This would be true if the algebraic variety $\ol{\M}$ would be normal at any $s \in \M$ (see \cite{GHMS09}, Corollary 12). From \cite{CaGiHeMaPa03}, Corollary 3 we deduce that the variety $\ol{\M}$ is definitely not normal. This leads us to the question how elimination polynomials should look like when the closure of the parameter domain is not normal. 

In order to elucidate this question we shall consider the following general situation. It turns out that the requirement that the coefficients of elimination polynomials should be geometrically robust constructible functions is quite natural. 

%doc2
\enter

Let $\varphi:V\to W$ be a finite surjective morphism of irreducible affine varieties $V$ and $W$ over $\C$ such that there exists a coordinate function $y\in\C[V]$ with $\C[V]=\C[W][y]$. Let $d:=[\C(V):\C(W)]$ be the degree of $\varphi$ and suppose that for any point $w\in W$ the cardinality of the fiber $\varphi^{-1}(w)$ is exactly $d$. Finally, let $Y$ be a new indeterminate and $F:=Y^d+\varphi_{d-1}Y^{d-1}+\dots+\varphi_0 \in\C(W)[Y]$ the minimal polynomial of $y$. Observe that coefficients of $F$, namely $\varphi_0,\dots,\varphi_{d-1}\in\C(W)$, are integral over $\C[W]$. We are now going to discuss a condition under which $F$ may be considered as an elimination polynomial. This condition will imply that $\varphi_0,\dots,\varphi_{d-1}$ are geometrically robust constructible functions.

We shall use the following abbreviations:\\
$A:=\C[W]$, $B:=A[\varphi_0,\dots,\varphi_{d-1}]$, $C:=A[y]$, $D:=B[y]$. We have the following commutative diagram of integral $\C$--algebra extensions:\\

%\begin{center}
%   \begin{psmatrix}[colsep=0.8,rowsep=0.1]
%                  & [name=B] $B$  &            \\
%     [name=A] $A$ &               & [name=D] $D$  \\
%                  & [name=C] $C$  &            \\
%  	 \ncline[nodesep=2pt]{->}{A}{B}
%		 \ncline[nodesep=2pt]{->}{A}{C}
%		 \ncline[nodesep=2pt]{->}{B}{D}
%		 \ncline[nodesep=2pt]{->}{C}{D}
%		 \end{psmatrix}
%\end{center}

\[
\begin{diagram}
\node{}  \node{B} \arrow{se}{} \node{}\\
\node{A} \arrow{se}{} \arrow{ne}{} \node{}  \node{D}\\
\node{}  \node{C} \arrow{ne}{}\node{}\\
\end{diagram}
\]
Observe that $D$ is isomorphic to $B[Y]/_{B[Y]\cdot F}$ and in particular a free $B$--module of rank $d$.

\begin{proposition}
Suppose that for any maximal ideal $\mathfrak{m}$ of $A$ the canonical $\C$--algebra homomorphism $A/_{\mathfrak{m}}\to C/_{\mathfrak{m}C}$ is unramified (\cite{Iversen73}, Chapter I) and that $\mathfrak{m}$ is contained in at most $d$ maximal ideals of $D$ (thus, intuitively, $F$ is an elimination polynomial). Then $\varphi_0,\dots,\varphi_{d-1}$ are geometrically robust constructible functions of $W$.
\end{proposition}

\begin{proof}

Let $\mathfrak{m}$ be an arbitrary maximal ideal of $A$. Since $A \to B$ is an integral ring extension, we deduce from Theorem--Definition \ref{theorem-definition} that it suffices to show that there exists a single maximal ideal $\mathfrak{N}$ of $B$ which contains $\mathfrak{m}$. Our assumptions yield a commutative diagram 

%\begin{center}
%   \begin{psmatrix}[colsep=0.8,rowsep=0.1]
%     [name=A] $A/_{\mathfrak{m}}$    &   &  [name=C] $C/_{\mathfrak{m}C}$ \\%
%	&	&	\\
%     [name=A2] $\C$                &   &  [name=C2] $\C^d$.       \\
%  	 \ncline[nodesep=2pt]{->}{A}{C}
%		 \ncline[nodesep=6pt]{->}{A2}{C2}		 \ncline[nodesep=10pt,linecolor=white]{A}{A2}\mput*{\rotatebox{270}{\scalebox{2}[1]{$\cong$}}}
%\ncline[nodesep=10pt,linecolor=white]{C}{C2}\mput*{\rotatebox{270}{\scalebox{2}[1]{$\cong$}}}
%		 \end{psmatrix}
%\end{center}

\[
\begin{diagram}
\node{A/_{\mathfrak{m}}} \arrow[2]{e}{}  \node{} \node{C/_{\mathfrak{m}C}}\\
\node{\rotatebox{270}{\scalebox{2}[1]{$\cong$}}}  \node{}  \node{\rotatebox{270}{\scalebox{2}[1]{$\cong$}}}\\
\node{\C} \arrow[2]{e}{}  \node{} \node{\C^d}\\
\end{diagram}
\]

Taking into account $C=A[y]$ we conclude that there exists a monic polynomial $G\in A[Y]$ of degree $d$ with discriminant $\rho\in A$ such that $C/_{\mathfrak{m}C}$ is isomorphic to $A[Y]$ divided by the ideal generated $\mathfrak{m}$ and $G$ and such that $\rho$ does not belong to $\mathfrak{m}$.

Let $\mathfrak{N}$ be an arbitrary maximal ideal of $B$ which contains $\mathfrak{m}$ and let $\ol{F}$ and $\ol{G}$ be the images of $F$ and $G$ in $B/_{\mathfrak{N}}[Y]$. Then we have $D/_{\mathfrak{N}D} \cong B/_{\mathfrak{N}}[Y]/_{B/_{\mathfrak{N}}[Y]\cdot \ol{F}}$ and therefore $\ol{F}$ divides $\ol{G}$ in $B/_{\mathfrak{N}}[Y]$. From $d=\text{deg}\ol{F}=\text{deg}\ol{G}$ and the fact that $F$ and $G$ are monic we deduce $\ol{F}=\ol{G}$. Since the discriminant $\rho$ of $G$ does not belong to $\mathfrak{m}$ we have $\rho \notin \mathfrak{N}$ and therefore the polynomial $\ol{F}$ is separable. Thus we obtain a commutative diagram 
\[
\begin{diagram}
\node{B/_{\mathfrak{N}}} \arrow[2]{e}{}  \node{} \node{D/_{\mathfrak{N}D}}\\
\node{\rotatebox{270}{\scalebox{2}[1]{$\cong$}}}  \node{}  \node{\rotatebox{270}{\scalebox{2}[1]{$\cong$}}}\\
\node{\C} \arrow[2]{e}{}  \node{} \node{\C^d}\\
\end{diagram}
\]

and in particular the canonical $\C$--algebra homomorphism $B/_{\mathfrak{N}} \to D/_{\mathfrak{N}D}$ is unramified. Hence the number of maximal ideals of $D$ which contain $\mathfrak{N}$ is exactly $d$. By assumption there are at most $d$ maximal ideals of $D$ containing $\mathfrak{m}$. Therefore any such ideal must contain $\mathfrak{N}$. Since $B\to D$ is an integral ring extension, we conclude that $\mathfrak{N}$ is the unique maximal ideal of $B$ which contains $\mathfrak{m}$. 
\end{proof}

   \section{A computation model with robust parameterized arithmetic circuits}
\label{A computation model with robust parameterized arithmetic circuits}

This section is devoted to a deeper understanding of the assumptions which lead to Theorem \ref{theorem model independent}. To this end we introduce a computation model which will be comprehensive enough to capture the essence of all known circuit based elimination algorithms in effective algebraic geometry and, mutatis mutandis, also of all other (linear algebra and truncated rewriting) elimination procedures (see \cite{Mora03}, \cite{Mora05}, and the references cited therein, and for truncated rewriting methods especially \cite{Dickens}).

The elimination problem and polynomial of Section \ref{independent of the model} were somewhat artificial. We shall show that the conclusions of Theorem \ref{theorem model independent} are still valid for much more natural elimination problems and polynomials if we restrict the notion of algorithm to the computation model we are going to introduce in this section.

The routines of our computation model will transform a given robust parameterized arithmetic (input) circuit into another robust parameterized arithmetic (output) circuit such that both circuits have the same parameter domain.

In the sequel, we shall use ordinary arithmetic circuits over $\C$ as \emph{generic computations} \cite{Burgisser97} (also called \emph{computation schemes} in \cite{Hei89}). The indegree zero nodes of these arithmetic circuits are labelled by scalars and parameter and input variables. 

The aim is to represent different parameterized arithmetic circuits of similar size and appearance by different specializations (i.e., instantiations) of the parameter variables in one and the same generic computation. For a suitable specialization of the parameter variables, the original parameterized arithmetic circuit may then be recovered by an appropriate reduction process applied to the specialized generic computation.

This alternative view of parameterized arithmetic circuits will be fundamental for the design of routines of the computation model we are going to describe in this section. The routines of our computation model will operate on robust parameterized arithmetic circuits and their basic ingredients will be subroutines which calculate parameter instances of suitable, by the model previously fixed, generic computations. These generic computations will be organized in finitely many families which will only depend on a constant number of discrete parameters. These discrete families constitute the basic building block of our computation model (see \cite{HKR11}, Section 3.2 for details about generic computation).

\enter

In the sequel we shall distinguish sharply between the notions of input variable and parameter and the corresponding categories of circuit nodes.

Input variables, called ``standard'', will occur in parameterized arithmetic circuits and generic computations. The input variables of generic computations will appear subdivided in three sorts, namely as ``parameter'', ``argument'' and ``standard'' input variables.

The computation model we are going to introduce now will assume different \emph{shapes}, each shape being determined by a finite number of a priori given \emph{discrete} (i.e., by tuples of natural numbers indexed) families of generic computations. The labels of the inputs of the ordinary arithmetic circuits which represent these generic computations will become subdivided into \emph{parameter}, \emph{argument} and \emph{standard} input variables. We shall use the letters like $U,U',U'',\dots$ and $W,W',W''$ to denote vectors of parameters, $Y,Y',Y'',\dots$ and $Z,Z',Z''$ to denote vectors of argument and $X,X',X'',\dots$ to denote vectors of standard input variables. 

We shall not write down explicitly the indexings of our generic computations by tuples of natural numbers. Generic computations will simply be distinguished by subscripts and superscripts, if necessary.

Ordinary arithmetic circuits of the form
\begin{center}
\begin{tabular}{l l l}
$R_{X_1}(W_{1};X^{(1)})$, 		& $R_{X_2}(W_{2};X^{(2)})$, 		& $\dots$ \\ 
$R_{X_1}'(W_{1'};X^{(1')})$, 	& $R_{X_2}'(W_{2'};X^{(2')})$, 	& $\dots$ \\ 
$\dots$ 											& $\dots$ 											& $\dots$ \\ 
\end{tabular}
\end{center}
represent a first type of a discrete family of generic computations (for each variable $X_1,X_2,\dots,X_n$ we suppose to have at least one generic computation). Other types of families of generic computations are of the form
\begin{center}
\begin{tabular}{l l l l}
$R_+(W;U,Y;X)$, 		& $R_+'(W';U',Y';X')$, 		& $R_+''(W'';U'',Y'';X'')$	 & $\dots$ \\ 
$R_{._/}(W;U,Y;X)$, & $R_{._/}'(W';U',Y';X')$, & $R_{._/}''(W'';U'',Y'';X'')$	 &$\dots$ \\ 
$R_{add}(W;Y,Z;X)$, & $R_{add}'(W';Y',Z';X')$, & $R_{add}''(W'';Y'',Z'';X'')$	 &$\dots$ \\ 
$R_{mult}(W;Y,Z;X)$,& $R_{mult}'(W';Y',Z';X')$,& $R_{mult}''(W'';Y'',Z'';X'')$	 & $\dots$ \\ 
\end{tabular}
\end{center}
and
\begin{center}
\begin{tabular}{l l l l}
$R_{div}(W;Y,Z;X)$, & $R_{div}'(W';Y',Z';X')$, & $R_{div}''(W'';Y'',Z'';X'')$	& $\dots$. \\ 
\end{tabular}
\end{center}

Here the subscripts refer to addition of, and multiplication or division by a parameter (or scalar) and to essential addition, multiplication and division. A final type of families of generic computations is of the form
$$R(W;Y;X),\ws R'(W';Y';X'),\ws R''(W'';Y'';X''),\dots$$ 
The inputs of the circuits handled by our computation model will only consist of standard variables.

From now on we have in mind a previously fixed shape when we refer to the computation model we are introducing. We start with a given finite set of discrete families of generic computations which constitute a shape as described before.

%falta algo?

\subsection[The notions of well behavedness under restrictions, isoparametricity and well behavedness under reductions]{The notions of well behavedness under restrictions,\\ isoparametricity and well behavedness under reductions}
\label{The notions of well behavedness under restrictions, isoparametricity and well behavedness under reductions}

A fundamental issue is how we recursively transform a given input circuit into another one with the same parameter domain. During such a transformation we make an iterative use of previously fixed generic computations. On their turn these determine the corresponding \emph{recursive routine} of our branching--free computation model.

We consider again our input circuit $\beta$. We suppose that we have already chosen for each node $\rho$, which depends at least on one of the input variables $X_1,\dots,X_n$, a generic computation 
$$R^{(\rho)}_{X_i}(W_{\rho};X^{(\rho)}),$$
$$R^{(\rho)}_{+}(W_{\rho};U_{\rho},Y_{\rho};X^{(\rho)}),$$
$$R^{(\rho)}_{._/}(W_{\rho};U_{\rho},Y_{\rho};X^{(\rho)}),$$
$$R^{(\rho)}_{add}(W_{\rho};Y_{\rho},Z_{\rho};X^{(\rho)}),$$
$$R^{(\rho)}_{mult}(W_{\rho};Y_{\rho},Z_{\rho};X^{(\rho)}),$$ 
$$R^{(\rho)}_{div}(W_{\rho};Y_{\rho},Z_{\rho};X^{(\rho)}),$$

and that this choice was made according to the label of $\rho$, namely $X_i, 1\leq i\leq n$, or addition of, or multiplication or division by an essential parameter, or essential addition, multiplication or division. Here we suppose that $U_{\rho}$ is a single variable, whereas $W_{\rho},Y_{\rho},Z_{\rho}$ and $X^{(\rho)}$ may be arbitrary vectors of variables.

Furthermore, we suppose that we have already precomputed for each node $\rho$ of $\beta$, which depends at least on one input, a vector $w_{\rho}$ of geometrically robust constructible functions defined on $\mathcal{M}$. If $\rho$ is an input node we assume that $w_{\rho}$ is a vector of complex numbers. Moreover, we assume that the length of $w_{\rho}$ equals the length of the variable vector $W_{\rho}$. We call the entries of $w_{\rho}$ the \emph{parameters at the node $\rho$} of the routine $\mathcal{A}$ applied to the input circuit $\beta$. 

We are now going to develop the routine $\mathcal{A}$ step by step. The routine $\mathcal{A}$ takes over all computations of $\beta$ which involve only parameter nodes, without modifying them. 

Consider an arbitrary internal node $\rho$ of $\beta$ which depends at least on one input. The node $\rho$ has two ingoing edges which come from two other nodes of $\beta$, say $\rho_1$ and $\rho_2$. Suppose that the routine $\mathcal{A}$, on input $\beta$, has already computed two results, namely $F_{\rho_1}$ and $F_{\rho_2}$, corresponding to the nodes $\rho_1$ and $\rho_2$. Suppose inductively that these results are vectors of polynomials depending on those standard input variables that occur in the vectors of the form $X^{(\rho')}$, where $\rho'$ is any predecessor node of $\rho$. Furthermore, we assume that the coefficients of these polynomials constitute the entries of a geometrically robust, constructible map defined on $\mathcal{M}$. Finally we suppose that the lengths of the vectors $F_{\rho_1}$ and $Y_{\rho}$ (or $U_{\rho}$) and $F_{\rho_2}$ and $Z_{\rho}$ coincide. 

The parameter vector $w_{\rho}$ of the routine $\mathcal{A}$ forms a geometrically robust, constructible map defined on $\mathcal{M}$, whose image we denote by $\mathcal{K}_{\rho}$. Observe that $\mathcal{K}_{\rho}$ is a constructible subset of the affine space of the same dimension as the length of the vectors $w_{\rho}$ and $W_{\rho}$. Denote by $\kappa_{\rho}$ the vector of the restrictions to $\mathcal{K}_{\rho}$ of the canonical projections of this affine space. We consider $\mathcal{K}_{\rho}$ as a new parameter domain with basic parameters $\kappa_{\rho}$. For the sake of simplicity we suppose that the node $\rho$ is labelled by an essential multiplication. Thus the corresponding generic computation has the form
%\begin{equation}
%\label{(1)}
%R^{(\rho)}_{._/}(W_{\rho};U_{\rho},Y_{\rho};X^{(\rho)})	
%\end{equation}
%or
\begin{equation}
\label{(2)}
R^{(\rho)}_{mult}(W_{\rho};Y_{\rho},Z_{\rho};X^{(\rho)}).	
\end{equation}

Let the specialized generic computation 
$$R_{mult}^{(\rho)}(\kappa_{\rho},Y_{\rho},Z_{\rho},X^{(\rho)})$$
be the by $\mathcal{K}_{\rho}$ parameterized arithmetic circuits obtained by substituting in the generic computation \eqref{(2)} for the vector of parameter variables $W_{\rho}$ the basic parameters $\kappa_{\rho}$. At the node $\rho$ we shall now make the following requirement on the routine $\mathcal{A}$ applied to the input circuit $\beta$:
\textit{\begin{enumerate}
	\item[(A)]  The by $\mathcal{K}_{\rho}$ parameterized arithmetic circuit of $R^{(\rho)}_{mult}(\kappa_{\rho};Y_{\rho},Z_{\rho};X^{(\rho)}),$ should be consistent and robust.
\end{enumerate}}
 
Observe that the requirement $(A)$ is automatically satisfied if all the generic computations of our shape are realized by totally division--free ordinary arithmetic circuits. 

Assume now that the routine $\mathcal{A}$ applied to the circuit $\beta$ satisfies the requirement $(A)$ at the node $\rho$ of $\beta$.

Recall that we assumed that the node $\rho$ is labelled by an essential multiplication and that the vectors $F_{\rho_1}$ and $Y_{\rho}$ and $F_{\rho_2}$ and $Z_{\rho}$ have the same length. Joining with the generic computation 
$$R^{(\rho)}_{mult}(W_{\rho};Y_{\rho},Z_{\rho};X^{(\rho)})$$ 
at $W_{\rho}, Y_{\rho}$ and $Z_{\rho}$ the previous computations of $w_{\rho}, F_{\rho_1}$ and $F_{\rho_2}$ we obtain a parameterized arithmetic circuit with parameter domain $\mathcal{M}$, whose final results are the entries of a vector which we denote by $F_{\rho}$. 

One deduces easily from our assumptions on $w_{\rho},F_{\rho_1}$ and $F_{\rho_2}$ and from the requirement $(A)$ in combination with Lemma \ref{lemma intermediate results} and Corollary \ref{proposition 1}, that the resulting parameterized arithmetic circuit is robust if it is consistent. The other possible labellings of the node $\rho$ by arithmetic operations are treated similarly. In particular, in case that $\rho$ is an input node labelled by the variable $X_i, 1\leq i\leq n$, the requirement $(A)$ implies that the ordinary arithmetic circuit $R^{(\rho)}_{X_i}(w_{\rho};X^{(\rho)})$ is consistent and robust and that all its intermediate results are polynomials in $X^{(\rho)}$ over $\C$ (although $R^{(\rho)}_{X_i}(w_{\rho};X^{(\rho)})$ may contain divisions). 

We call the recursive routine $\mathcal{A}$ (on input $\beta$) \emph{well behaved under restrictions} if the requirement $(A)$ is satisfied at any node $\rho$ of $\beta$ which depends at least on one input and if joining the corresponding generic computation with $w_{\rho}$, $F_{\rho_1}$ and $F_{\rho_2}$ produces a consistent circuit (observe that this last condition is automatically satisfied when the specialized generic computation of $(A)$ is essentially division--free). If the routine $\mathcal{A}$ is well behaved under restrictions, then $\mathcal{A}$ transforms step by step the input circuit $\beta$ into another consistent robust arithmetic circuit, namely $\mathcal{A}(\beta)$, with parameter domain $\mathcal{M}$. 

As a consequence of the recursive structure of $\mathcal{A}(\beta)$, each node $\rho$ of $\beta$ generates a subcircuit of $\mathcal{A}(\beta)$ which we call the component of $\mathcal{A}(\beta)$ generated by $\rho$. The output nodes of each component of $\mathcal{A}(\beta)$ form the hypernodes of a hypergraph $\mathcal{H}_{\mathcal{A}(\beta)}$ whose hyperedges are given by the paths connecting the nodes of $\mathcal{A}(\beta)$ contained in distinct hypernodes of $\mathcal{H}_{\mathcal{A}(\beta)}$. The hypergraph $\mathcal{H}_{\mathcal{A}(\beta)}$ may be shrunk to the DAG structure of $\beta$ and therefore we denote the hypernodes of $\mathcal{H}_{\mathcal{A}(\beta)}$ in the same way as the nodes of $\beta$. Notice that well behavedness under restrictions is in fact a property which concerns the hypergraph $\mathcal{H}_{\mathcal{A}(\beta)}$. 

We call $\mathcal{A}$ a (recursive) \emph{parameter routine} if $\mathcal{A}$ does not introduce new standard variables. In the previous recursive construction of the routine $\mathcal{A}$, the parameters at the nodes of $\beta$, used for the realization of the circuit $\mathcal{A}(\beta)$, are supposed to be generated by recursive parameter routines.

We are now going to consider another requirement of our recursive routine $\mathcal{A}$, which will lead us to the notion of \emph{isoparametricity} of $\mathcal{A}$. 

Let us turn back to the previous situation at the node $\rho$ of the input circuit $\beta$. Notations and assumptions will be the same as before. From Lemma \ref{lemma intermediate results} we deduce that the intermediate result of $\beta$ associated with the node $\rho$, say $G_{\rho}$, is a polynomial in $X_1,\dots,X_n$ whose coefficients form the entries of a geometrically robust, constructible map defined on $\mathcal{M}$, say $\theta_{\rho}$. Let $\mathcal{T}_{\rho}$ be the image of this map and observe that $\mathcal{T}_{\rho}$ is a constructible subset of a suitable affine space. The intermediate results of the circuit $\mathcal{A}(\beta)$ at the elements of the hypernode $\rho$ of $\mathcal{H}_{\mathcal{A}(\beta)}$ constitute a polynomial vector which we denote by $F_{\rho}$.

We shall now make another requirement on the routine $\mathcal{A}$ at the node $\rho$ of $\beta$.%at the node $\rho$ on the routine $\mathcal{A}$ applied to the input circuit $\beta$:
\textit{\begin{enumerate}
	\item[(B)] There exists a geometrically robust, constructible map $\sigma_{\rho}$ defined on $\mathcal{T}_{\rho}$ such that $\sigma_{\rho}\circ \theta_{\rho}$ constitutes the coefficient vector of $F_{\rho}$.
\end{enumerate}}

In view of the comments made in \cite{HKR11}, Section 3.3.1 we call the recursive routine $\mathcal{A}$ \emph{isoparametric} (on input $\beta$) if requirements $(A)$ and $(B)$ are satisfied at any node $\rho$ of $\beta$ which depends at least on one input. 

\enter

Suppose again that the recursive routine $\mathcal{A}$ is well behaved under restrictions. We call $\mathcal{A}$ \emph{well behaved under reductions} (on input $\beta$) if $\mathcal{A}(\beta)$ satisfies the following requirement:

\begin{quote}
\textit{Let $\rho$ and $\rho'$ be distinct nodes of $\beta$ which compute the same intermediate results. Then the intermediate results at the hypernodes $\rho$ and $\rho'$ of $\mathcal{H}_{\mathcal{A}(\beta)}$ are identical. Mutatis mutandis the same is true for the computation of the parameters of $\mathcal{A}$ at any node of $\beta$.
}	
\end{quote}

Assume that the routine $\mathcal{A}$ is recursive and well behaved under reductions. One verifies then easily that, taking into account the hypergraph structure $\mathcal{H}_{\mathcal{A}(\beta)}$ of $\mathcal{A}(\beta)$, any reduction procedure on $\beta$ may canonically be extended to a reduction procedure of $\mathcal{A}(\beta)$.

It can be shown that under very light assumptions well behavedness under reductions implies isoparametricity. Since from the point of view of software architecture well behavedness under reductions is a well motivated quality attribute of recursive routines we see that isoparametricity is computationally meaningful concept. For details we refer to \cite{HKR11}, Section 3.3.2.3.

\subsection{Operations with recursive routines}

Let $\mathcal{A}$ and $\mathcal{B}$ be recursive routines as before and suppose that they are well behaved under restrictions and isoparametric or even well behaved under reductions. Assume that $\mathcal{A}(\beta)$ is an admissible input for $\mathcal{B}$. We define the composed routine $\mathcal{B}\circ\mathcal{A}$ in such a way that $(\mathcal{B}\circ\mathcal{A})(\beta)$ becomes the parameterized arithmetic circuit $\mathcal{B}(\mathcal{A}(\beta))$. Since the routines $\mathcal{A}$ and $\mathcal{B}$ are well behaved under restrictions, we see easily that $(\mathcal{B}\circ\mathcal{A})(\beta)$ is a consistent, robust parameterized arithmetic circuit with parameter domain $\mathcal{M}$. From Lemma \ref{lemma intermediate results} and Corollary \ref{proposition 1} we deduce that $\mathcal{B}\circ\mathcal{A}$ is a isoparametric recursive routine if $\mathcal{A}$ and $\mathcal{B}$ are isoparametric. In case that $\mathcal{A}$ and $\mathcal{B}$ are well behaved under reductions, one verifies immediately that $\mathcal{B}\circ\mathcal{A}$ is also well behaved under reductions. Therefore, under these assumptions, we shall consider $\mathcal{B}\circ\mathcal{A}$ also as a routine of our computation model. 
  
The identity routine is trivially well behaved under restrictions and reductions and in particular isoparametric.

Let $\mathcal{A}$ and $\mathcal{B}$ be two isoparametric recursive routines of our computation model. Assume that the robust parameterized arithmetic circuit $\beta$ is an admissible input for $\mathcal{A}$ and $\mathcal{B}$ and that there is given a one--to--one correspondence $\lambda$ which identifies the output nodes of $\mathcal{A}(\beta)$ with the input nodes of $\mathcal{B}(\beta)$. Often, for a given input circuit $\beta$, the correspondence $\lambda$ is clear by the context. If we limit ourselves to input circuits $\beta$ where this occurs, we obtain from $\mathcal{A}$ and $\mathcal{B}$ a new routine, called their \emph{join}, which transforms the input circuit $\beta$ into the output circuit $\mathcal{B}(\beta)*_{\lambda}\mathcal{A}(\beta)$ (here we suppose that $\mathcal{B}(\beta)*_{\lambda}\mathcal{A}(\beta)$ is consistent). Analyzing now $\mathcal{B}(\beta)*_{\lambda}\mathcal{A}(\beta)$, we see that the join of $\mathcal{A}$ with $\mathcal{B}$ is well behaved under restrictions in the most obvious sense. Since by assumption the routines $\mathcal{A}$ and $\mathcal{B}$ are recursive, the circuits $\mathcal{A}(\beta)$ and $\mathcal{B}(\beta)$ inherit from $\beta$ a superstructure given by the hypergraphs $\mathcal{H}_{\mathcal{A}(\beta)}$ and $\mathcal{H}_{\mathcal{B}(\beta)}$. Analyzing again this situation, we see that any reduction procedure on $\beta$ can be extended in a canonical way to the circuit $\mathcal{B}(\beta)*_{\lambda}\mathcal{A}(\beta)$. This means that the join of $\mathcal{A}$ with $\mathcal{B}$ is also well behaved under reductions if the same is true for $\mathcal{A}$ and $\mathcal{B}$. More caution is at order with the notion of isoparametricity. In a strict sense, the join of two isoparametric recursive routines $\mathcal{A}$ and $\mathcal{B}$ is not necessarily isoparametric. However, condition $(B)$ is still satisfied between the output nodes of $\beta$ and $\mathcal{B}(\beta)*_{\lambda}\mathcal{A}(\beta)$. A routine with this property is called \emph{output isoparametric}.

The \emph{union} of the routines $\mathcal{A}$ and $\mathcal{B}$ assigns to the input circuit $\beta$ the juxtaposition of $\mathcal{A}(\beta)$ and $\mathcal{B}(\beta)$. Thus, on input $\beta$, the final results of the union of $\mathcal{A}$ and $\mathcal{B}$ are the final results of $\mathcal{A}(\beta)$ and $\mathcal{B}(\beta)$ (taken separately in case of ambiguity). The union of $\mathcal{A}$ and $\mathcal{B}$ behaves well under restrictions and reductions and is isoparametric if the same is true for $\mathcal{A}$ and $\mathcal{B}$. 

Observe also that for a recursive routine $\mathcal{A}$ which behaves well under restrictions and reductions the following holds: let $\beta$ be a robust parameterized arithmetic circuit that broadcasts to a circuit $\beta^*$ and assume that $\beta$ and $\beta^*$ are admissible circuits for $\mathcal{A}$. Then $\mathcal{A}(\beta)$ broadcasts to $\mathcal{A}(\beta^*)$.

\subsection{Elementary routines}

From these considerations we conclude that routines, constructed as before by iterated applications of the operations isoparametric recursion, composition, join and union, are still, in a suitable sense, well behaved under restrictions and output isoparametric. If only recursive routines become involved that behave well under reductions, we may also allow broadcastings at the interface of two such operations.

This remains true when we introduce, as we shall do now, in our computational model the following additional type of routine construction.

Let $\beta$ be the robust, parameterized circuit considered before, and let $R(W;Y;X)$ be a generic computation belonging to our shape list. Let $w_{\beta}$ be a precomputed vector of geometrically robust constructible functions with domain of definition $\M$ and suppose that $w_{\beta}$ and $W$ have the same vector length and that the entries of $w_{\beta}$ are the final results of an output isoparametric parameter routine applied to the circuit $\beta$. Moreover suppose that the final results of $\beta$ form a vector of the same length as $Y$.

Let $\mathcal{K}$ be the image of $w_{\beta}$. Observe that $\mathcal{K}$ is a constructible subset of the affine space which has the same dimension as the vector length of $W$. Denote by $\kappa$ the vector of the restrictions to $\mathcal{K}$ of the canonical projections of this affine space. We denote by $R(\kappa;Y;X)$ the ordinary arithmetic circuit over $\C$ obtained by substituting in the generic computation $R(W;Y;X)$ the vector of parameter variables $W$ by $\kappa$. We shall now make the following requirement: 
\textit{
\begin{enumerate}
	\item[(C)] The ordinary arithmetic circuit $R(\kappa;Y;X)$ should be consistent and robust.
\end{enumerate}
}
Observe that requirement $(C)$ is obsolete when $R(W;Y;X)$ is a totally division--free ordinary arithmetic circuit.

Suppose now that requirement $(C)$ is satisfied. A new routine, say $\mathcal{B}$, is obtained in the following way: on input $\beta$ the routine $\mathcal{B}$ joins with the generic computation $R(W;Y;X)$ at $W$ and $Y$ the previous computation of $w_{\beta}$ and the circuit $\beta$. 

From Lemma \ref{lemma intermediate results} and Corollary \ref{proposition 1} we deduce that the resulting parameterized arithmetic circuit $\mathcal{B}(\beta)$ has parameter domain $\mathcal{M}$ and is robust if it is consistent. We shall therefore require that $\mathcal{B}(\beta)$ is consistent (this condition is automatically satisfied if $R(\kappa;Y;X)$ is essentially division--free). One sees immediately that the routine $\mathcal{B}$ is well behaved under restrictions and reductions and is output isoparametric.

From now on we shall always suppose that all our recursive routines are isoparametric and that requirement $(C)$ is satisfied when we apply this last type of routine construction.%56

An \emph{elementary routine} of our computation model is finally obtained by the iterated application of all these construction patterns, in particular the last one, isoparametric recursion, composition, join and union. As far as only recursion becomes involved that is well behaved under reductions, we allow also broadcastings and reductions at the interface of two constructions. Of course, the identity routine belongs also to our model. The set of all these routines is therefore closed under these constructions and operations. 

We call an elementary routine \emph{essentially division--free} if it admits as input only essentially division--free, robust parameterized arithmetic circuits and all specialized generic computations used to compose it are essentially division--free. The outputs of essentially division--free elementary routines are always essentially division--free robust circuits. The set of all essentially division--free elementary routines is also closed under the mentioned constructions and operations. 

We have seen that elementary routines are, in a suitable sense, well behaved under restrictions. In the following statement we formulate explicitly the property of an elementary routine to be output isoparametric. This will be fundamental in our subsequent complexity considerations.

\begin{proposition}\label{prop: T subset and composition}
Let $\mathcal{A}$ be an elementary routine of our computation model. Then $\mathcal{A}$ is output isoparametric. More explicitly, let $\beta$ be a robust, parameterized arithmetic circuit with parameter domain $\mathcal{M}$. Suppose that $\beta$ is an admissible input for $\mathcal{A}$. Let $\theta$ be a geometrically robust, constructible map defined on $\mathcal{M}$ such that $\theta$ represents the coefficient vector of the final results of $\beta$ and let $\mathcal{T}$ be the image of $\theta$. Then $\mathcal{T}$ is a constructible subset of a suitable affine space and there exists a geometrically robust, constructible map $\sigma$ defined on $\mathcal{T}$ such that the composition map $\sigma\circ\theta$ represents the coefficient vector of the final results of $\mathcal{A}(\beta)$.
\end{proposition}%57

In case that $\mathcal{A}$ is a recursive routine, Proposition \ref{prop: T subset and composition} expresses nothing but the requirement $(B)$ applied to the output nodes of $\beta$.

\subsection{End of the description of our computation model}

Elementary routines do not contain branchings. In order to capture the whole spectrum of all known circuit based elimination methods in effective algebraic geometry, we are now going to introduce the concepts of \emph{algorithm} and \emph{procedure} admitting some limited branchings. This finishes the description of our computation model. We shall proceed rather informally.

An algorithm will be a dynamic DAG of elementary routines which will be interpreted as pipes. At the end points of the pipes, decisions may be taken which depend on testing the validity of suitable universally quantified Boolean combinations of equalities between robust constructible functions defined on the parameter domain under consideration. The output of such an \emph{equality test} is a bit vector which determines the next elementary routine (i.e., pipe) to be applied to the output circuit produced by the preceding elementary routine (pipe). This gives rise to a computation model which contains branchings. These branchings depend on a limited type of decisions at the level of the underlying abstract data type, namely the mentioned equality tests. Because of this limitation of branchings, we shall call the algorithms of our model \emph{branching parsimonious} (compare \cite{GH01} and \cite{CaGiHeMaPa03}). A branching parsimonious algorithm $\mathcal{A}$ which accepts a robust parameterized arithmetic circuit $\beta$ with parameter domain $\mathcal{M}$ as input produces a new robust circuit $\mathcal{A}(\beta)$ with parameter domain $\mathcal{M}$. In particular $\mathcal{A}(\beta)$ \emph{does not contain any branchings}. 

Recall that our two main constructions of elementary routines depend on a previous selection of generic computations from a given shape list. This selection may be handled by calculations with the indexing of the shape list. We shall think that these calculations become realized by deterministic Turing machines. At the beginning, for a given robust parametric input circuit $\beta$ with parameter domain $\mathcal{M}$, a tuple of fixed (i.e., of $\beta$ independent) length of natural numbers is determined. This tuple constitutes an initial configuration of a Turing machine computation which determines the generic computations of our shape list that intervene in the elementary routine under construction. The entries of this tuple of natural numbers are called \emph{invariants} of the circuit $\beta$. These invariants, whose values may also be Boolean (i.e., realized by the natural numbers $0$ or $1$), depend mainly on algebraic or geometric properties of the final results of $\beta$. However, they may also depend on structural properties of the labelled DAG $\beta$.%60

For example, the invariants of $\beta$ may express that $\beta$ has $r$ parameters, $n$ inputs and outputs, (over $\C$) non--scalar size and depth at most $L$ and $l$, that $\beta$ is totally division--free, that the final results of $\beta$ have degree at most $d\leq 2^{l}$ and that for any parameter instance their specializations form a reduced regular sequence in $\C[X_1,\dots,X_n]$, where $X_1,\dots,X_n$ are the inputs of $\beta$.

Some of these invariants (e.g., the syntactical ones like number of parameters, inputs and outputs and non--scalar size and depth) may simply be read--off from the labelled DAG structure of $\beta$. Others, like the truth value of the statement that the specializations of the final results of $\beta$ at any parameter instance form a reduced regular sequence, have to be precomputed by an elimination algorithm from a previously given software library in effective commutative algebra or algebraic geometry or their value has to be fixed in advance as a precondition for the elementary routine which becomes applied to $\beta$.

In the same vein we may equip any elementary routine $\mathcal{A}$ with a Turing computable function which from the values of the invariants of a given input circuit $\beta$ decides whether $\beta$ is admissible for $\mathcal{A}$, and, if this is the case, determines the generic computations of our shape list which intervene in the application of $\mathcal{A}$ to $\beta$. %61

We shall now go a step further letting depend the internal structure of the computation on the circuit $\beta$. In the simplest case this means that we admit that the vector of invariants of $\beta$, denoted by $\text{inv}(\beta)$, determines the architecture of a first elementary routine, say $\mathcal{A}_{\text{inv}(\beta)}$, which admits $\beta$ as input. Observe that the architectures of the elementary routines of our computation model may be characterized by tuples of fixed length of natural numbers. We consider this characterization as an \emph{indexing} of the elementary routines of our computation model. We may now use this indexing in order to combine dynamically elementary routines by composition, join and union. Let us restrict our attention to the case of composition. In this case the output circuit of one elementary routine is the input for the next routine. The elementary routines which compose this display become implemented as pipes which start with a robust input circuit and end with a robust output circuit. Given such a pipe and an input circuit $\gamma$ for the elementary routine $\mathcal{B}$ representing the pipe, we may apply suitable equality tests to the final results of $\mathcal{B}(\gamma)$ in order to determine a bit vector which we use to compute the index of the next elementary routine (seen as a new pipe) which will be applied to $\mathcal{B}(\gamma)$ as input.

A \emph{low level program} of our extended computation model is now a text, namely the transition table of a deterministic Turing machine, which computes a function $\psi$ realizing the following tasks.

Let as before $\beta$ be a robust parameterized arithmetic circuit. Then $\psi$ returns first on input $\text{inv}(\beta)$ a Boolean value, zero or one, where one is interpreted as the informal statement ``$\beta$ is an admissible input''. If this is the case, then $\psi$ returns on $\text{inv}(\beta)$ the index of an elementary routine, say $\mathcal{A}_{\text{inv}(\beta)}$, which admits $\beta$ as input. Then $\psi$ determines the equality tests which have to be realized with the final results of $\mathcal{A}_{\text{inv}(\beta)}(\beta)$. Depending on the outcome of these equality tests $\psi$ determines an index value corresponding to a new elementary routine which admits $\mathcal{A}_{\text{inv}(\beta)}(\beta)$ as input. Continuing in this way one obtains as end result an elementary routine $\mathcal{A}^{(\beta)}$, which applied to $\beta$, produces a final output circuit $\mathcal{A}^{(\beta)}(\beta)$. The function $\psi$ represents all these index computations. We denote by $\psi(\beta)$ the \emph{dynamic} vector of all data computed by $\psi$ on input $\beta$. 

The \emph{algorithm} represented by $\psi$ is the partial map between robust parametric arithmetic circuits that assigns to each admissible input $\beta$ the circuit $\mathcal{A}^{(\beta)}(\beta)$ as output. Observe that elementary routines are particular algorithms. This kind of algorithms constitute our \emph{computation model}. We remark that any algorithm of this model is \emph{output isoparametric}. If the pipes of an algorithm are all represented by essentially division--free elementary routines, we call the algorithm itself \emph{essentially division--free}.

One sees easily that the ``Kronecker algorithm'' \cite{GLS01} (compare also \cite{Giusti1}, \cite{Giusti2} and \cite{Giusti3}) for solving non--degenerate polynomial equation systems over the complex numbers may be programmed in our extended computation model. Observe that the Kronecker algorithm requires more than a single elementary routine for its design. In order to understand this, recall that the Kronecker algorithm accepts as input an ordinary division--free arithmetic circuit which represents by its output nodes a reduced regular sequence of polynomials $G_1,\dots,G_n$ belonging to $\C[X_1,\dots,X_n]$. In their turn, the polynomials $G_1,\dots,G_n$ determine a \emph{degree pattern}, say $\Delta:=(\delta_1,\dots,\delta_n)$, with $\delta_i:=\deg \{G_1=0,\dots,G_i=0 \}$ for $1\leq i\leq n$. 

After putting the variables $X_1,\dots,X_n$ in generic position with respect to $G_1,\dots,$ $G_n$, the algorithm performs $n$ recursive steps to eliminate them, one after the other. Finally the Kronecker algorithm produces an ordinary arithmetic circuit which computes the coefficients of $n+1$ univariate polynomials $P,V_1,\dots,V_n$ over $\C$. These polynomials constitute a ``geometric solution'' (see \cite{GLS01}) of the equation system $G_1=0,\dots,G_n=0$ because they represent the zero dimensional algebraic variety $V:=\left\{ G_1=0,\dots,G_n=0 \right\}$ in the following ``parameterized'' form:
$$V:=\left\{ (V_1(t),\dots,V_n(t));t\in\C,P(t)=0 \right\}.$$ 
Let $\beta$ be any robust, parameterized arithmetic circuit with the same number of inputs and outputs, say $X_1,\dots,X_n$ and $G_1(U,X_1,\dots,X_n),\dots,G_n(U,X_1,\dots,X_n)$, respectively. Suppose that the parameter domain of $\beta$, say $\mathcal{M}$, is irreducible and that $\text{inv}(\beta)$ expresses that for each parameter instance $u\in\mathcal{M}$ the polynomials $G_1(u,X_1,\dots,X_n),\dots,G_n(u,X_1,\dots,X_n)$ form a reduced regular sequence in $\C[X_1,\dots,X_n]$ with fixed (i.e., from $u\in\mathcal{M}$ independent) degree pattern. Suppose, furthermore, that the degrees of the individual polynomials $G_1(u,X_1,\dots,X_n)$, $\dots$, $G_n(u,X_1,\dots,X_n)$ are also fixed and that the variables $X_1,\dots,X_n$ are in generic position with respect to the varieties $\{ G_1(u,X)=0,\dots,G_i(u,X)=0 \}, 1\leq i\leq n$. Then, on input $\beta$, the Kronecker algorithm runs a certain number (which depends on $\Delta$) 
of elementary routines of our computation model which finally become combined by consistent iterative joins until the desired output is produced.

We say that a given algorithm $\mathcal{A}$ of our extended model \emph{computes} (only) \emph{parameters} if $\mathcal{A}$ satisfies the following condition:
\begin{quote}
\emph{for any admissible input $\beta$ the final results of $\mathcal{A}(\beta)$ are all parameters.}	
\end{quote}

Suppose that $\mathcal{A}$ is such an algorithm and $\beta$ is the robust parametric arithmetic circuit with parameter domain $\mathcal{M}$ which we have considered before. Observe that $\mathcal{A}(\beta)$ contains the input variables $X_1,\dots,X_n$ and that possibly new variables, which we call \emph{auxiliary}, become introduced during the execution of the algorithm $\mathcal{A}$ on input $\beta$. Since the algorithm $\mathcal{A}$ computes only parameters, the input and auxiliary variables become finally eliminated by the application of recursive parameter routines and evaluations. We may therefore \emph{collect garbage} in order to reduce $\mathcal{A}(\beta)$ to a \emph{final output circuit} $\mathcal{A}_{\text{final}}(\beta)$ whose intermediate results are only parameters.

If we consider the algorithm $\mathcal{A}$ as a partial map which assigns to each admissible input circuit $\beta$ its final output circuit $\mathcal{A}_{\text{final}}(\beta)$, we call $\mathcal{A}$ a \emph{procedure}.

In this case, if $\psi$ is a low level program defining $\mathcal{A}$, we call $\psi$ a \emph{low level procedure program}. 

A particular feature of our extended computation model is the following:\\
there exists a non--negative integer $f$ (depending on the recursion depth of $\mathcal{A}$) and non--decreasing real valued functions $C_f \geq 0$ ,\dots, $C_0 \geq 0$ depending on one and the same dynamic integer vector, such that with the previous notations and $L_{\beta}$, $L_{\mathcal{A}(\beta)}$ denoting the non--scalar sizes of the circuits $\beta$ and $\mathcal{A}(\beta)$ the condition
$$L_{\mathcal{A}(\beta)} \leq C_f(\psi(\beta))L^f_{\beta} +\dots+ C_0(\psi(\beta))$$
is satisfied.

In the case of the Kronecker algorithm (and most other elimination algorithms of effective algebraic geometry) we have $f:=1$, because the recursion depth of the basic routines which intervene is one.

\subsection{Procedures}\ws\ws

In the sequel we shall need a particular variant of the notion of a procedure which enables us to capture the following situation.

Suppose we have to find a computational solution for a formally specified general algorithmic problem and that the formulation of the problem depends on certain parameter variables, say $U_1,\dots,U_r$, input variables, say $X_1,\dots,X_n$ and output variables, say $Y_1,\dots,Y_s$. Let such a problem formulation be given and suppose that its input is implemented by the robust parameterized arithmetic circuit $\beta$ considered before, interpreting the parameter variables $U_1,\dots,U_r$ as the basic parameters $\pi_1,\dots,\pi_r$. 

Then an algorithm $\mathcal{A}$ of our extended computation model which \emph{solves} the given algorithmic problem should satisfy the architectural requirement we are going to describe now.%65

The algorithm $\mathcal{A}$ should be the composition of two subalgorithms $\mathcal{A}^{(1)}$ and $\mathcal{A}^{(2)}$ of our computation model which satisfy on input $\beta$ the following conditions:
\textit{\begin{enumerate}
	\item[(i)] The subalgorithm $\mathcal{A}^{(1)}$ computes only parameters, $\beta$ is admissible for $\mathcal{A}^{(1)}$ and none of the indeterminates $Y_1,\dots,Y_s$ is introduced in $\mathcal{A}^{(1)}(\beta)$ as auxiliary variable (all other auxiliary variables become eliminated during the execution of the subalgorithm $\mathcal{A}^{(1)}$ on the input circuit $\beta$).
	\item[(ii)] The circuit $\mathcal{A}_{\text{final}}^{(1)}(\beta)$ is an admissible input for the subalgorithm $\mathcal{A}^{(2)}$, the indeterminates $Y_1,\dots,Y_s$ occur as auxiliary variables in $\mathcal{A}^{(2)}(\mathcal{A}_{\text{final}}^{(1)}(\beta))$ and the final results of $\mathcal{A}^{(2)}(\mathcal{A}_{\text{final}}^{(1)}(\beta))$ depend only on $\pi_1,\dots,\pi_r$ and $Y_1,\dots,Y_s$.	 
\end{enumerate}}

To the circuit $\mathcal{A}^{(2)}(\mathcal{A}_{\text{final}}^{(1)}(\beta))$ we may, as in the case when we compute only parameters, apply garbage collection. In this manner $\mathcal{A}^{(2)}(\mathcal{A}_{\text{final}}^{(1)}(\beta))$ becomes reduced to a final output circuit $\mathcal{A}_{\text{final}}(\beta)$ with parameter domain $\mathcal{M}$ which contains only the inputs $Y_1,\dots,Y_s$.

Observe that the subalgorithm $\mathcal{A}^{(1)}$ is by Proposition \ref{prop: T subset and composition} an output isoparametric procedure of our extended computation model (the same is also true for the subalgorithm $\mathcal{A}^{(2)}$, but this will not be relevant in the sequel).

We consider the algorithm $\mathcal{A}$, as well as the subalgorithms $\mathcal{A}^{(1)}$ and $\mathcal{A}^{(2)}$, as \emph{procedures} of our extended computation model. In case that the \emph{subprocedures} $\mathcal{A}^{(1)}$ and $\mathcal{A}^{(2)}$ are essentially division--free, we call also the procedure $\mathcal{A}$ \emph{essentially division--free}. 

The architectural requirement given by conditions $(i)$ and $(ii)$ may be interpreted as follows:\\
the subprocedure $\mathcal{A}^{(1)}$ is a pipeline which transmits only parameters to the subprocedure $\mathcal{A}^{(2)}$. In particular, no (true) polynomial is transmitted from $\mathcal{A}^{(1)}$ to $\mathcal{A}^{(2)}$. 	

Nevertheless, let us observe that on input $\beta$ the procedure $\mathcal{A}$ establishes by means of the underlying low level program $\psi$ an additional link between $\beta$ and the subprocedure $\mathcal{A}^{(2)}$ applied to the input $\mathcal{A}^{(1)}(\beta)$. The elementary routines which constitute $\mathcal{A}^{(2)}$ on input $\mathcal{A}^{(1)}(\beta)$ become determined by index computations which realizes $\psi$ on $\text{inv}(\beta)$ and certain equality tests between the intermediate results of $\mathcal{A}^{(1)}(\beta)$. In this sense the subprocedure $\mathcal{A}^{(1)}$ transmits not only parameters to the subprocedure but also a limited amount of digital information which stems from the input circuit $\beta$.   	 

The decomposition of the procedure $\mathcal{A}$ into two subprocedures $\mathcal{A}^{(1)}$ and $\mathcal{A}^{(2)}$ satisfying conditions $(i)$ and $(ii)$ represents an architectural restriction which is justified when it makes sense to require that on input $\beta$ the number of essential additions and multiplications contained in $\mathcal{A}_{\text{final}}(\beta)$ is bounded by a function which depends only on $\text{inv}(\beta)$. This is the case in the example of the next subsection where a substantial use of this restriction is made (see \cite{HKR11}, Section 4.1, Observations). 

\subsection{A hard elimination problem}
\label{A hard elimination problem}

Let $n\in\N$ and $S_1,\dots,S_n,T,U_1,\dots,U_n$ and $X_1,\dots,X_n$ be indeterminates. Let $U:=(U_1,\dots,U_n)$, $S:=(S_1,\dots,S_n)$, $X:=(X_1,\dots,X_n)$ and $G_1^{(n)}:=X_1^2-X_1-S_1,\dots,G_n^{(n)}:=X_n^2-X_n-S_n$, $H^{(n)}:=\sum_{1 \leq i \leq n} 2^{i-1}X_i + T \prod_{1 \leq i \leq n} (1+ (U_i-1)X_i)$.

Observe that the polynomials $G_1^{(n)},\dots,G_n^{(n)}$ form a reduced regular sequence in $\C[S,T,U,X]$ and that they define a subvariety $V_n$ of the affine space $\A^n\times\A^1\times\A^n\times\A^n$ which is isomorphic to $\A^n\times\A^1\times\A^n$ and hence irreducible and of dimension $2n+1$. Moreover, the morphism $V_n \to \A^n\times\A^1\times\A^n $ which associates to any point $(s,t,u,x)\in V_n$ the point $(s,t,u)$, is finite and generically unramified. Therefore the morphism $\pi_n:V_n \to \A^n\times\A^1\times\A^n\times\A^1$ which associates to any $(s,t,u,x)\in V_n$ the point $(s,t,u,H^{(n)}(t,u,x))\in \A^n\times\A^1\times\A^n\times\A^1$ is finite and its image $\pi_n(V_n)$ is a hypersurface of $\A^n\times\A^1\times\A^n\times\A^1$ with irreducible minimal equation $F^{(n)}\in \C[S,T,U,Y]$.

Thus, $F^{(n)}$ is an irreducible elimination polynomial of degree $2^n$. Therefore any equation of $\C[S,T,U,Y]$ which defines $\pi_n(V_n)$ in $\A^n\times\A^1\times\A^n\times\A^1$ is up to a scalar factor a power of $F^{(n)}$.

The equations $G_1^{(n)}=0,\dots,G_n^{(n)}=0$ and the polynomial $H^{(n)}$ represent a so called \emph{flat family of zero--dimensional elimination problems} with associated elimination polynomial $F^{(n)}$ (see \cite{HKR11}, Section 4.1 for the notion of a flat family of zero--dimensional elimination problems).

We consider now $S_1,\dots,S_n,T,U_1,\dots,U_n$ as basic parameters, $X_1,\dots,X_n$ as input and $Y$ as output variables.

Let $\mathcal{A}$ be an essentially division--free procedure of our extended computation model satisfying the following condition: \\
$\mathcal{A}$ accepts as input any robust parameterized arithmetic circuit $\beta$ which represents a flat family of zero--dimensional elimination problem with associated elimination polynomial $F$ and $\mathcal{A}_{\text{final}}(\beta)$ has a single input $Y$ and a single final result which defines the same hypersurface as $F$.

With this notions and notations we have the following result.

\begin{theorem}
\label{proposition A}
There exist an ordinary division--free arithmetic circuit $\beta_n$ of size $O(n)$ over $\C$ with inputs $S_1,\dots,S_n$, $T$, $U_1,\dots,U_n$, $X_1,\dots,X_n$ and final results $G_1^{(n)},\dots,G_n^{(n)},H^{(n)}$. The essentially division--free, robust parameterized arithmetic circuit $\gamma_n:=\mathcal{A}_{\text{final}}(\beta_n)$ depends on the basic parameters $S_1,\dots,S_n$, $T$, $U_1,\dots,U_n$ and the input $Y$ and its single final result is a power of $F^{(n)}$. The circuit $\gamma_n$ performs at least $\Omega(2^{\frac{n}{2}})$ essential multiplications and at least $\Omega(2^n)$ multiplications with parameters. As ordinary arithmetic circuit over $\C$ with inputs $S_1,\dots,S_n$, $T$, $U_1,\dots,U_n$ and $Y$, the circuit $\gamma_n$ has non--scalar size at least $\Omega(2^n)$.
\end{theorem}

The proof of Theorem \ref{proposition A} is similar as that of Theorem \ref{theorem model independent}. Moreover, Theorem \ref{proposition A} implies the asymptotic optimality of the Kronecker algorithm within our computation model. For details we refer the reader to \cite{HKR11}, Section 4 where also other examples of elimination problems are exhibited which are hard for algorithms of our computation model.

%-------------------------------------------------------------------------------
\section{Approximative computations}\label{sec:Approximative computations}
%-------------------------------------------------------------------------------
We are now going to apply the algorithmic model of Section \ref{A computation model with robust parameterized arithmetic circuits} in a different context, namely that of approximative computations.

Let $\beta$ be a robust parameterized arithmetic circuit with parameter domain $\mathcal{M}$, basic parameters $\pi_1,\dots,\pi_r$ and inputs $X_1,\dots,X_n$. Let $U_1,\dots,U_r$ be parameter variables, $U:=(U_1,\dots,U_r)$, $\pi:=(\pi_1,\dots,\pi_r)$, $X:=(X_1,\dots,X_n)$ and suppose that $\beta$ is essentially division-free and has a single final result $G$.

In this section we are going to introduce the notion of an \emph{approximative $\beta$--computation} and to discuss how in our computation model an approximative $\beta$--computation can be transformed in an exact one.

Let $\mathfrak{a}$ be the vanishing ideal of $\ol{\mathcal{M}}$ in $\C[U]$ and let us fix a polynomial $P\in\C[U]$ such that $\ol{\mathcal{M}}_{P}$ is a Zariski open and dense subset of $\mathcal{M}$. Let $\epsilon$ be a new indeterminate.

\begin{definition}[Approximative parameter instance]\label{def: Approximative parameter instance}
An approximative parameter instance for $\beta$ is a vector $u(\epsilon)= (u_1(\epsilon),\dots,u_r(\epsilon)) \in
\C{((\epsilon))}^r$ which constitutes a meromorphic map germ at the origin such that any polynomial of $\mathfrak{a}$ vanishes at $u(\epsilon)$ and $P(u(\epsilon))\neq 0$ holds.
\end{definition}
\mydefinitions{\label{def: Approximative parameter instance} Approximative parameter instance}

Let $u(\epsilon)$ be an approximative parameter instance for $\beta$. Then there exists an open disc $\Delta$ around $0$ such that for any complex number $c\in A - \{ 0 \}$ the germ $u(\epsilon)$ is holomorphic at $c$ and such that $P(u(c))\neq 0$ holds. This implies that any polynomial of $\mathfrak{a}$ vanishes at $u(c)$ and that in particular $u(c)$ belongs to $\mathcal{M}$.

For technical reasons we need the following result.

\begin{lemma}
\label{lemma open disc}
Let notations and assumptions be as before. Let $\phi:\mathcal{M}\to\A^m$ be a geometrically robust constructible map and let $u(\epsilon)$ be an approximative parameter instance for $\beta$. Then there exists an open disc $\Delta$ of $\C$ around the origin and a germ $\psi$ of meromorphic functions at the origin such that $u(\epsilon)$ and $\psi$ are holomorphic on $\Delta - \{ 0 \}$ and such that any complex number $c\in\Delta - \{ 0 \}$ satisfies the conditions $P(u(c))\neq 0$ and $\psi(c)=\phi(u(c))$. 
\end{lemma}

\begin{proof}
There exists an open disc $\Delta'$ of $\C$ around the origin such that $u(\epsilon)$ is everywhere defined on $\Delta' - \{ 0 \}$ and such that any $c \in \Delta' - \{ 0 \}$ satisfies the condition $P(u(c))\neq 0$. Let $\mathcal{N}$ be the Zariski closure of the image of $\Delta' - \{ 0 \}$ under $u(\epsilon)$. There exists a Zariski open and dense subset $\mathcal{U}$ of $\mathcal{N}$ with $\mathcal{U}\subset\mathcal{M}$ such that $\phi$ is rational and everywhere defined on $\mathcal{U}$. Moreover there exists a non--zero polynomial $Q\in\C[U]$ such that $\mathcal{N}_Q$ is contained in $\mathcal{U}$ and Zariski dense in $\mathcal{N}$. Therefore there exists a complex number $c_0\in \Delta'-\{ 0 \}$ with $Q(u(c_0))\neq 0$. Hence the set $K:=\{ c\in \Delta'-\{ 0 \}; P(u(c))=0 \}$ is finite. Thus we may chose an open disc $\Delta$ around the origin with $\Delta \subset \Delta'$ and $\Delta \cap K =\emptyset$ such that the image of $\Delta$ under $u(\epsilon)$ is contained in $\mathcal{N}_Q$. We conclude now that $u(\epsilon)$ is everywhere defined on $\Delta-\{ 0 \}$ and that every $c\in\Delta - \{ 0 \}$ satisfies the condition $u(c)\in\mathcal{U}$. For $c\in\Delta-\{ 0 \}$ let $\psi(c):=\phi(u(c))$. Then $\psi:\Delta-\{ 0 \} \to \C^m$ is a well defined meromorphic function. Let $c\in\Delta-\{ 0 \}$ and let $(c_k)_{k\in\N}$ be a sequence of complex numbers $c_k\in\Delta -\{ 0 \}$ which converges to $c$. Then $(u(c_k))_{k\in\N}$ is a sequence of points $u(c_k)\in\mathcal{U}$ which converges to $u(c)$ and hence the sequence $(\phi(u(c_k)))_{k\in\N}$ converges to $\phi(u(c))$ and is therefore bounded. This implies that $\psi$ is holomorphic at $c$. Thus $\psi$ is holomorphic on $\Delta -\{ 0 \}$ and for any $c\in \Delta -\{ 0 \}$ we have $\psi(c)=\phi(u(c))$. Therefore we may interpret $\psi$ as a meromorphic map germ at the origin. 
\end{proof}

\ws

Let $u(\epsilon)$ be an approximative parameter instance. Then following Lemma \ref{lemma open disc} there exists an open disc $\Delta$ of $\C$ around the origin such that for any node $\rho$ of $\beta$ with intermediate result $G_{\rho}(\pi,X)$ the expression $G_{\rho}(u(\epsilon),X)$ defines a polynomial in $X_1,\dots,X_n$ whose coefficients are well determined meromorphic functions on $\Delta$ which are holomorphic on $\Delta-\{ 0 \}$. We denote by $\beta^{(u(\epsilon))}$ the labelled DAG of $\beta$ where we assign to each node $\rho$ of $\beta$ the polynomial $G_{\rho}(u(\epsilon),X)$. We call $\beta^{(u(\epsilon))}$ an \emph{approximative $\beta$--computation} and denote by $G^{(u(\epsilon))}$ the final result of $\beta^{(u(\epsilon))}$. 

We say that the approximative $\beta$--computation $\beta^{(u(\epsilon))}$ represents the polynomial $H\in\C[X]$ if there exists a polynomial $H^{(u(\epsilon))}\in\C[[\epsilon]][X]$ whose coefficient vector with respect to $X$ constitutes a germ of functions which are holomorphic at the origin such that the final result $G^{(u(\epsilon))}$ of $\beta^{(u(\epsilon))}$ can be written as $G^{(u(\epsilon))}= H+\epsilon H^{(u(\epsilon))}$.

%Let $u=(u_1(\epsilon),\dots,u_r(\epsilon))$ be an approximative parameter instance for $\beta$. Then we may interpret $u(\epsilon)$ as a map from $\A^1$ to $\A^r$ which is holomorphic at the origin and whose image is contained in $\mathcal{M}$, except for finitely many arguments. In the arguments the images belong to $\ol{\mathcal{M}}$. Since $\beta$ is a robust parameterized arithmetic circuit we see that $\beta^{(u(\epsilon))}$ is consistent.

Let $W_{\beta}$ be the set of coefficient vectors of the final results of the ordinary arithmetic circuits $\beta^{(u)}$ in $\C[X]$, where $u\in\mathcal{M}$. We consider $W_{\beta}$ as a constructible subset of a suitable affine ambient space. Consequently the (Zariski or strong) closure $\ol{W}_{\beta}$ of $W_{\beta}$ in its ambient space is well defined. In view of the next result the expression ``approximative $\beta$--computation'' becomes selfexplanatory.

\begin{theorem}
\label{th: equivalent conditions}
Let notations and assumptions be as before and let $H\in\C[X]$. Then the following three conditions are equivalent
\begin{itemize}
	\item[$(i)$] there exists an approximative $\beta$--computation that represents $H$. 
	\item[$(ii)$] there exists a sequence $(u_k)_{k\in\N}$ with $u_k\in\mathcal{M}$ such that the final results of the sequence $(\beta^{(u_k)})_{k\in\N}$ of ordinary circuits converge to $H$ in $\C[X]$.  
	\item[$(iii)$] the coefficient vector of $H$ belongs to $\ol{W}_{\beta}$. 
\end{itemize}
\end{theorem} 
\begin{proof}
The conditions $(ii)$ and $(iii)$ are obviously equivalent because one is only a restatement of the other. It suffices therefore to show the implications $(i)\Rightarrow (ii)$ and $(iii)\Rightarrow (i)$. We first prove $(i)\Rightarrow (ii)$.

Suppose that there exists an approximative parameter instance $u(\epsilon)$ for $\beta$ such that $\beta^{(u(\epsilon))}$ represents $H\in\C[X]$ by means of a polynomial $H^{(u(\epsilon))}\in\C[[\epsilon]][X]$ whose coefficient vector constitutes a holomorphic map germ at the origin. Thus $H+\epsilon H^{u(\epsilon)}$ is the final result of $\beta^{(u(\epsilon))}$. We may choose a sequence $(\epsilon_k)_{k\in\N}$ of non--zero complex numbers such that for any $k\in\N$ the germ $u(\epsilon)$ is holomorphic at $\epsilon_k$ and satisfies the condition $P(u(\epsilon_k))\neq 0$ and such that $(\epsilon_k)_{k\in\N}$ converges to zero. Without loss of generality we may suppose that the coefficients of $H^{(u(\epsilon))}$ are holomorphic at $\epsilon_k$ for any $k\in\N$.

Let $u_k:=u(\epsilon_k)$. Then $u_k$ belongs to $\mathcal{M}$ and $H+\epsilon_k H^{(u(\epsilon))} (\epsilon_k,X)$ is the final result of the ordinary arithmetic circuit $\beta^{(u_k)}$. Moreover the sequence $(H+\epsilon_k H^{u(\epsilon)}(\epsilon_k,X))_{k\in\N}$ converges to $H$. Therefore condition $(ii)$ is satisfied.

The other implication is more cumbersome. We adapt the argumentation of \cite{Alder84} (see also \cite{Lickteig90}, \S A) to our computation model. Let us now prove $(iii)\Rightarrow (i)$.

Without loss of generality we may suppose that $\ol{\mathcal{M}}$ is an irreducible affine variety. Let $\theta$ be the geometrically robust constructible map which assigns to each $u\in\mathcal{M}$ the coefficient vector of the final result $G^{(u)}$ of the circuit $\beta^{(u)}$. Observe that $\theta(\ol{\mathcal{M}}_{P})$ is Zariski dense in $\ol{W}_{\beta}$. Thus $B:=\ol{\ol{W}_{\beta}-\theta(\ol{\mathcal{M}}_P)}$ is a proper Zariski closed subset of $\ol{W}_{\beta}$. Let $q$ be the dimension of the irreducible affine variety $\ol{W}_{\beta}$ and let $h$ be the coefficient vector of the polynomial $H$. By assumption we have $h\in\ol{W}_{\beta}$. In case of $q=0$ we conclude $h\in W_{\beta}$ and we are finished. Therefore we may suppose without loss of generality $q>0$.

%pag 5  
By Noether's Normalization Lemma there exists a surjective finite morphism of irreducible affine varieties $\lambda:\ol{W}_{\beta}\to\A^q$. Since $B$ is a proper Zariski closed subset of $\ol{W}_{\beta}$ we have $\lambda(B)\subsetneqq  \A^q$. We may therefore choose a point $z\in\A^q - \lambda(B)$. Let $L$ be a straight line of $\A^r$ which passes through $\lambda(h)$ and $z$. Then $\lambda(h)$ belongs to $L$ and $\lambda(B)\cap L$ is a finite set. Since the morphism $\lambda$ is finite, the irreducible components of $\lambda^{-1}(L)$ are all closed subcurves of $\ol{W}_{\beta}$ which become mapped onto $L$ by $\lambda$. In particular there exists an irreducible component $C$ of $\lambda^{-1}(L)$ which contains $h$. Since $\lambda(B)\cap L$ is finite we have $C\nsubseteq B$. Therefore $C\cap B$ is finite too. Suppose that $\theta(\ol{\mathcal{M}}_{P}) \cap C$ is finite. Then we may conclude that $C \cap B$ is infinite, a contradiction. Therefore $\theta(\ol{\mathcal{M}}_P) \cap C$ is infinite. Hence the Zariski closure of $\theta^{-1}(C)$ in $\ol{\mathcal{M}}$ contains an irreducible component $V$ such that $\theta(V_P)$ is Zariski dense in $C$. Let $q^*:=\dim V$ and $u \in V_P$. Observe $q^* >0$ and that $B^*:=\ol{\theta^{-1}(\theta(u))} \cap V$ is a proper Zariski closed subset of $V$. Again by Noether's Normalization Lemma there exists a surjective finite morphism of irreducible affine varieties $\lambda^*:V \to \A^{q^*}$.

Since $B^*$ is a proper Zariski closed subset of $V$ we have $\lambda^*(B^*) \subsetneqq \A^{q^*}$. Therefore we may choose again a point $z^* \in\A^{q^*} - \lambda^*(B^*)$ and a straight line $L^*$ of $\A^{q^*}$ which passes through $\lambda^*(u)$ and $z^*$. Thus $\lambda^*(B^*) \cap L^*$ is a finite set and $\lambda^*(u)$ belongs to $L^*$. The irreducible components of $(\lambda^*)^{-1}(L^*)$ are all closed subcurves of $V$ which become mapped onto $L^*$ by $\lambda^*$. In particular there exists an irreducible component $C^*$ of $(\lambda^*)^{-1}(L^*)$ which contains $u$. Since $\lambda^*(B^*) \cap L^*$ is finite we conclude $C^* \nsubseteq B^*$. Moreover from $C^*$ irreducible, $u\in C^*$ and $u\in V_P$ we conclude that $(C^*)_P$ is Zariski dense in $C^*$. Hence $C^* \nsubseteq B_*$ implies that there exists a point $u^*\in(C^*)_P - B^*$. Thus we have $\theta(u^*) \neq \theta(u)$. Therefore $\theta((C^*)_P)$ is Zariski dense in $C$.

In this way we have found two irreducible closed subcurves $C^*$ of $\ol{\mathcal{M}}$ and $C$ of $\ol{W}_{\beta}$ with $(C^*)_P$ nonempty such that $\theta$ maps $(C^*)_P$ into $C$ and $\theta((C^*)_P)$ is Zariski dense in $C$. Moreover we have $h\in C$.

The restriction of $\theta$ to $(C^*)_P$ is again a geometrically robust constructible map. Therefore $\theta$ induces a finite field extension $\C(C) \subset \C(C^*)$.

Let $\mathcal{R}$ be the integral closure of $\C[C]$ in $\C(C^*)$. Then $\mathfrak{m}$ may be extended to a maximal ideal $\mathfrak{M}$ of $\mathcal{R}$. Observe that $\mathcal{R}$ is finitely generated over $\C$. Thus $\mathcal{R}$ is the coordinate ring of a non--singular curve and $\mathfrak{M}$ defines a point of this curve. Thus $\mathcal{R}_{\mathfrak{M}}$ is a regular $\C$--algebra of dimension one and therefore there exists an embedding of $\mathcal{R}_{\mathfrak{M}}$ into the power series ring $\C[[\epsilon]]$ which maps any generator of the maximal ideal of $\mathcal{R}_{\mathfrak{M}}$ onto a power series of order one. Moreover the elements of $\mathcal{R}_{\mathfrak{M}}$ become mapped onto power series which constitute holomorphic function germs at the origin. Hence the coordinate functions of $C^*$, given by the restrictions of $\pi_1,\dots,\pi_r$ to $C^*$, can be represented by Laurent series $u_1(\epsilon),\dots,u_r(\epsilon)$ of $\C((\epsilon))$ which constitute meromorphic function germs at the origin. Let $u(\epsilon):=(u_1(\epsilon),\dots,u_r(\epsilon))$. Then we have $P(u(\epsilon))\neq 0$ and $\theta(u(\epsilon))$ is by Lemma \ref{lemma open disc} a well defined meromorphic map germ at the origin. Moreover $\theta(u(\epsilon))$ belongs to $\C[C]$ and hence to $\mathcal{R}_{\mathfrak{M}}$. This implies that the entries of the vector $\theta(u(\epsilon))$ are power series of $\C[[\epsilon]]$ which constitute holomorphic function germs at the origin. Moreover $h - \theta(u(\epsilon))$ belongs to the maximal ideal of $\C[C]_{\mathfrak{m}}$ and hence to that of $\mathcal{R}_{\mathfrak{M}}$. This means that $\epsilon$ divides the entries of $h - \theta(u(\epsilon))$ in $\C[[\epsilon]]$. Observe now that $\theta(u(\epsilon))$ is the coefficient vector of the final result $G^{(u(\epsilon))}$ of $\beta^{(u(\epsilon))}$. We conclude now that there exists a polynomial $H^{(u(\epsilon))}\in\C[[\epsilon]][X]$, whose coefficients constitute holomorphic function germs at the origin, such that $H + \epsilon H^{(u(\epsilon))} = G^{(u(\epsilon))}$ holds. Since $C^*$ is contained in $\ol{\mathcal{M}}$ we have finally $A(u(\epsilon))=0$ for any polynomial $A\in\mathfrak{a}$.
\end{proof}

\ws

Theorem \ref{th: equivalent conditions} is the technical main result of \cite{Alder84}, where an analogous statement is proved in the particular case $\mathcal{M}:=\A^r$ (see also \cite{Lickteig90}, \S A).

In the sequel we retake the notations of the proof of Theorem \ref{th: equivalent conditions}.

Let $u(\epsilon)$ be an approximative parameter instance for $\beta$ such that $\beta^{(u(\epsilon))}$ represents a polynomial $H\in\C[X]$ whose coefficient vector $h$ belongs to $W_{\beta}$. Then, using interpolation, one may first compute the coefficient vector $\theta(u(\epsilon))$ of the final result $G^{(u(\epsilon))}$ of $\beta^{u(\epsilon)}$ and then recompute $G^{(u(\epsilon))}$ from the entries of $\theta(u(\epsilon))$ and the indeterminates $X_1,\dots,X_n$. We may compute $H$ in the same way from the entries of $h=\lim_{\epsilon\to 0}\theta(u(\epsilon))$ and $X_1,\dots,X_n$. This algorithm may be interpreted as two essentially division--free procedures $\mathcal{A}^{(1)}$ and $\mathcal{A}^{(2)}$ of our extended computation model where $\mathcal{A}^{(1)}$ becomes applied to the input circuit $\beta^{(u(\epsilon))}$ and returns the parameter vector $\theta(u(\epsilon))$ and $\mathcal{A}^{(2)}$ becomes applied to the input $\theta(u(\epsilon))=\mathcal{A}^{(1)}(\beta^{(u(\epsilon))})$ and returns $G^{(u(\epsilon))}$. Thus $\mathcal{A}^{(1)}$ computes on input $\beta^{(u(\epsilon))}$ only parameters. Therefore $\mathcal{A}:=(\mathcal{A}^{(1)},\mathcal{A}^{(2)})$ may be interpreted as an essentially division--free procedure of our extended computation model which accepts $\beta^{(u(\epsilon))}$ as input and returns $G^{(u(\epsilon))}$ as output (observe that $X_1,\dots,X_n$ play simultaneously the r\^ole of input and output variables). Suppose now more generally that there is given an essentially division--free procedure $\mathcal{A}=(\mathcal{A}^{(1)},\mathcal{A}^{(2)})$ of our extended computation model which accepts $\beta$ as input and returns as output a robust parameterized arithmetic circuit $\mathcal{A}_{\text{final}}(\beta)$ whose single final result is the polynomial $G$. 

Like in Section \ref{A computation model with robust parameterized arithmetic circuits}
 we consider now the circuit $\gamma:=\mathcal{A}_{\text{final}}(\beta)$ and the geometrically robust constructible map $\nu:=\mathcal{A}^{(1)}(\beta)$ whose domain of definition is $\mathcal{M}$. Observe that $\gamma$ is an essentially division--free, robust arithmetic circuit with parameter domain $\mathcal{M}$. Let $\mathcal{S}$ be the image of the geometrically robust constructible map $\nu$. Then $\mathcal{S}$ is a constructible subset of a suitable affine space $\A^p$ and there exists a geometrically robust constructible map $\psi:\mathcal{S}\to\A^m$ and a vector of $m$--variate polynomials $\omega^*$ such that for $\nu^*:=\psi\circ\nu$ the geometrically robust constructible map $\omega^*(\nu^*)=\omega^*\circ\nu^*$ is the coefficient vector $\theta$ of the final result $G$ of $\beta$. Observe that $\mathcal{S}^*:=\psi(\mathcal{S}):=\nu^*(\mathcal{M})$ is a constructible subset of $\A^m$ and that $\omega^*(\mathcal{S}^*)=\omega^*(\nu^*(\mathcal{M}))=\theta(\mathcal{M})=W_{\beta}$ holds. Thus we may interpret $\omega^*$ as a surjective polynomial map $\omega^*:\mathcal{S}^*\to W_{\beta}$. In the terminology of \cite{CaGiHeMaPa03}, Section 3.2, the constructible set $\mathcal{S}^*$ is a \emph{data structure}, $W_{\beta}$ an (abstract) \emph{object class} and $\omega^*$ a \emph{holomorphic encoding} of $W_{\beta}$ by $\mathcal{S}^*$, where $W_{\beta}$ represents the set of polynomials $\{ G^{(u)}; u\in\mathcal{M} \}$. Moreover $m$ is called the \emph{size} of the data structure $\mathcal{S}$. In an analogous way we may also consider $\mathcal{S}$ as a data structure and $\omega:=\omega^*\circ\psi$ an encoding of $W_{\beta}$ (observe that $W_{\beta}=\theta(\mathcal{M})=\omega^*\circ\psi\circ\nu(\mathcal{M})=\omega(\mathcal{S})$ holds). Again in the terminology of \cite{CaGiHeMaPa03}, Section 3.2, $\omega$ is a \emph{continuous encoding} of the object class $W_{\beta}$ by the data structure $\mathcal{S}$ whose size is $p$.  

\begin{lemma}
\label{lemma three sequences}
Let $(u_k)_{k\in\N}$, $u_k\in\mathcal{M}$, $(s_k)_{k\in\N}$, $s_k\in\mathcal{S}$ and $(s_k^*)_{k\in\N}$, $s^*_k\in\mathcal{S}^*$ three sequences which satisfy the following condition:
\begin{quote}
the sequences $(\theta(u_k))_{k\in\N}$, $(\omega(s_k))_{k\in\N}$ and $(\omega^*(s_k^*))_{k\in\N}$ converge each to a point of $W_{\beta}$.	
\end{quote}
Then the sequences $(\nu(u_k))_{k\in\N}$, $(s_k)_{k\in\N}$ and $(\nu^*(u_k))_{k\in\N}$, $(s_k^*)_{k\in\N}$ converge each to a point of $\mathcal{S}$ or $\mathcal{S}^*$, respectively.
\end{lemma}

\begin{proof}
Observe that the procedure $\mathcal{A}^{(1)}$ is output isoparametric. Therefore there exists a geometrically robust constructible map $\sigma$ with domain of definition $W_{\beta}$ such that $\sigma\circ\theta=\nu$ holds. Since $\sigma$ and $\theta$ are continuous with respect to the strong topologies of their source and image spaces and $(\theta(u_k))_{k\in\N}$ converges to a point $g$ of $W_{\beta}$, the sequence $(\nu(u_k))_{k\in\N}$ converges to $\sigma(g)$. Hence there exists a parameter instance $u\in\mathcal{M}$ with $g=\theta(u)$. Thus we have $\sigma(g)=(\sigma\circ\theta)(u)=\nu(u)$ and $\nu(u)$ belongs to $\mathcal{S}$. This implies that the sequence $(\nu(u_k))_{k\in\N}$ converges to a point of $\mathcal{S}$. 

Observe that for each $k\in\N$ there exists a parameter instance $u_k\in\mathcal{M}$ with $s_k=\nu(u_k)$. Therefore we have 
$$\omega(s_k)=\omega\circ\nu(u_k)=(\omega^*\circ\psi\circ\nu)(u_k)=(\omega^*\circ\nu^*)(u_k)=\theta(u_k).$$
By assumption the sequence $(\omega(s_k))_{k\in\N}$, and hence also the sequence $(\theta(u_k))_{k\in\N}$, converge to a point of $W_{\beta}$. As we have shown, this implies that the sequence $(\nu(u_k))_{k\in\N}$, i.e. the sequence $(s_k)_{k\in\N}$ converges to a point of $\mathcal{S}$, namely to $\sigma(g)$.

The statements about the sequences $(\nu^*(u_k))_{k\in\N}$ and $(s_k^*)_{k\in\N}$ are shown in the same way.
\end{proof}

\ws

Lemma \ref{lemma three sequences} implies that $\omega$ and $\omega^*$ are \emph{robust encodings} in the terminology of \cite{CaGiHeMaPa03}, Section 3.2.

Let now be given an approximative parameter instance $u(\epsilon)$ of $\beta$ such that $\beta^{(u(\epsilon))}$ represents a polynomial $H$ whose coefficient vector belongs to $W_{\beta}$.  

We wish to evaluate the polynomial $H$. There exists a parameter instance $u\in\mathcal{M}$ with $H=G^{(u)}$, but we are in general not able to find it. However we may be able to find a sequence $(\epsilon_k)_{k\in\N}$ of non zero complex numbers such that for any $k\in\N$ the germ $u(\epsilon)$ is holomorphic at $\epsilon_k$ and satisfies the condition $P(u(\epsilon_k))\neq 0$. This implies that $u_k:=u(\epsilon_k)$ belongs to $\mathcal{M}$. Since the approximative $\beta$--computation $\beta^{(u(\epsilon))}$ represents the polynomial $H$, the sequence $(G^{(u_k)})_{k\in\N}$ converges to $H$. From Lemma \ref{lemma three sequences} we deduce now that the sequences $(\nu(u_k))_{k\in\N}$ and $(\nu^*(u_k))_{k\in\N}$ converge to points $s$ and $s^*$ of $\mathcal{S}$ and $\mathcal{S}^*$, respectively. Therefore $\omega(s)$ and $\omega^*(s^*)$ are equal and form the coefficient vector of $H$. In case that the essential parameters of the circuit $\gamma=\mathcal{A}_{\text{final}}(\beta)$ constitute the entries of the geometrically robust constructible map $\nu^*$, we may reinterpret $\gamma$ as a robust parameterized arithmetic circuit with parameter domain $\mathcal{S}^*$. In this interpretation $\gamma^{(s^*)}$ becomes an ordinary division--free arithmetic circuit in $\C[X]$ whose single final result is $H$.

If the coefficient vector of $H$ belongs to $\ol{W_{\beta}}$, but not to $W_{\beta}$ we need the additional requirement that the sequences $(\nu(u(\epsilon_k)))_{k\in\N}$ or $(\nu^*(u(\epsilon_k)))_{k\in\N}$ are convergent (or at least bounded) in order to make the same kind of conclusions. However this depends on the technical condition that the geometrically robust constructible map $\sigma$ of the proof of Lemma \ref{lemma three sequences} may continuously be extended to the closure $\ol{W_{\beta}}$ of $W_{\beta}$ or on the choice of the sequence $(\epsilon_k)_{k\in\N}$.   

%TODO
In order to illustrate these technical considerations we are now going to discuss the following example.

Let $L,n$ be natural numbers, $r:=(L+n+1)^2$, $K_{L,n}:=4(L+n+1)^2+2$ and let $U_1,\dots,U_r$ be parameter and $X_1,\dots,X_n$ input variables. Let $\mathcal{M}_{L,n}:=\A^r$ and $\pi_1,\dots,\pi_r$ the canonical projections of $\mathcal{M}_{L,n}$ onto $\A^1$.

Following \cite{Burgisser97}, Chapter 9, Exercise 9.18 there exists a totally division--free generic computation $\beta_{L,n}$ with a single final result $G_{L,n}$ such that any polynomial $H\in\C[X_1,\dots,X_n]$ is evaluable by at most $L$ essential multiplications if and only if there exists a parameter instance $u\in\mathcal{M}_{L,n}$ such that $H=G_{L,n}^{(u)}$ holds. As a totally division--free generic computation we may interpret $\beta_{L,n}$ as a robust parameterized arithmetic circuit with parameter domain $\mathcal{M}_{L,n}$, basic parameters $\pi_1,\dots,\pi_r$ and inputs $X_1,\dots,X_n$.

Let us fix points $\gamma_1,\dots,\gamma_{K_{L,n}}$ of $\Z^n$ such that $(\gamma_1,\dots,\gamma_{K_{L,n}})$ is an identification sequence for the set $W_{L,n}:=\{ G_{L,n}^{(u)};u\in\mathcal{M}_{L,n} \}$. This means that for any $H_1,H_2 \in W_{L,n}$ the implication $H_1(\gamma_1)=H_2(\gamma_1),\dots,H_1(\gamma_{K_{L,n}})=H_2(\gamma_{K_{L,n}})\Rightarrow H_1=H_2$ holds. 

There exist identification sequences $(\gamma_1,\dots,\gamma_{K_{L,n}})$ of bit length at most $4(L+1)$ (see \cite{CaGiHeMaPa03} and \cite{GHMS09} for details).

As before let $\mathcal{A}=(\mathcal{A}^{(1)},\mathcal{A}^{(2)})$ be an essentially division--free procedure of our computation model which on input $\beta_{L,n}$ returns a robust parameterized arithmetic circuit $\mathcal{A}_{\text{final}}(\beta_{L,n})$ whose single final result is $G_{L,n}$. Let $\nu_{L,n}:=\mathcal{A}^{(1)}(\beta_{L,n})$ and let $\mathcal{S}_{L,n}$ be the image of the geometrically robust constructible map $\nu_{L,n}$. We think $\mathcal{S}_{L,n}$ as a constructible subset of an affine space $\A^{p_{L,n}}$. Furthermore let be given a geometrically robust constructible map $\psi_{L,n}:\mathcal{S}_{L,n}\to\A^{m_{L,n}}$ and a vector of $m_{L,n}$--variate polynomials $\omega^*_{L,n}$ such that for $\nu^*_{L,n}:=\psi_{L,n}\circ\nu_{L,n}$ the vector of coefficients of $G_{L,n}$ with respect to the variables $X_1,\dots,X_n$ can be written as $\omega^*_{L,n}\circ\nu^*_{L,n}$.

Assume now $p_{L,n}=K_{L,n}$ and that $\nu_{L,n}:\mathcal{M}_{L,n}\to\mathcal{S}_{L,n}$ is the polynomial (and hence geometrically robust constructible) map which for $u\in\mathcal{M}_{L,n}$ is defined by $\nu_{L,n}(u):=(G^{(u)}(\gamma_1),\dots,G^{(u)}(\gamma_{m_{L,n}}))$.

Let $\mathcal{A}^{(2)}$ be a procedure of our computation model which accepts $\nu_{L,n}=\mathcal{A}(\beta_{L,n})$ as input and returns on $\nu_{L,n}$ as output the polynomial $G_{L,n}$ (such procedures always exist). Then the size $p_{L,n}=K_{L,n}$ of the continuous encoding $\omega^*_{L,n}\circ\psi_{L,n}:\mathcal{S}_{L,n}\to W_{\beta_{L,n}}$ of $W_{\beta_{L,n}}$ is $K_{L,n}=4(L+n+1)^2+2$ whereas for $\mathcal{S}^*_{L,n}:=\psi_{L,n}(\mathcal{S})=\nu^*_{L,n}(\mathcal{M})$ the map $\omega^*_{L,n}:\mathcal{S}^*_{L,n}\to W_{\beta_{L,n}}$ represents a holomorphic encoding of $W_{\beta_{L,n}}$ of size $m_{L,n}=2^{\Omega(Ln)}$ (see \cite{GHMS09}, Theorem 23). The proof of this lower bound may be adapted to show the estimate $m_{L,n}=2^{\Omega(Ln)}$ for any procedure $\mathcal{A}=(\mathcal{A}^{(1)},\mathcal{A}^{(2)})$ of our extended computation model which recalculates the polynomial $G_{L,n}$ from the circuit $\beta_{L,n}$.

Thus, there are natural classes of polynomials which have continuous encodings of ``small size'' whereas their holomorphic encodings may become necessarily ``large''.

%-----------------------------------------------
%%%%%%%%%%%%%% Bibliography %%%%%%%%%%%%%%%%%%%%
%-----------------------------------------------

%\bibliographystyle{amsalpha} 
%\bibliographystyle{splncs03}

%\bibliographystyle{myalpha}
%\bibliography{simple}

\begin{thebibliography}{GHM{\etalchar{+}}98}

\bibitem[Ald84]{Alder84}
A.~Alder.
\newblock {\em {G}renzrang und {G}renzkomplexit{\"{a}}t aus algebraischer und
  topologischer {S}icht}.
\newblock PhD thesis, Universit{\"{a}}t Z{\"{u}}rich, {P}hilosophische
  {F}akult{\"{a}}t II, 1984.

\bibitem[BCS97]{Burgisser97}
P.~B{\"{u}}rgisser, M.~Clausen, M.~A. Shokrollahi.
\newblock {\em Algebraic Complexity Theory. {G}rundlehren der mathematischen
  {W}issenschaften},  315.
\newblock Springer Verlag, 1997.

\bibitem[CGH89]{Caniglia89}
L.~Caniglia, A.~Galligo, J.~Heintz.
\newblock Some new effectivity bounds in computational geometry.
\newblock {\em Applied Algebra, Algebraic Algorithms and Error Correcting
  Codes. Proc. of the 6th Intern. Conference AAECC, Best Paper Award AAECC-6.
  Springer LNCS}, 357 131--151, 1989.

\bibitem[CGH{\etalchar{+}}03]{CaGiHeMaPa03}
D.~Castro, M.~Giusti, J.~Heintz, G.~Matera, L.~M. Pardo.
\newblock The hardness of polynomial equation solving.
\newblock {\em Foundations of Computational Mathematics}, 3 (4) 347--420, 2003.

\bibitem[DFGS91]{Dickens}
A.~Dickenstein, N.~Fitchas, M.~Giusti, C.~Sessa.
\newblock The membership problem of unmixed polynomial ideals is solvable in
  single exponential time.
\newblock {\em Discrete Applied Mathematics}, 33 73--94, 1991.

\bibitem[FIK86]{FIK86}
T.~Freeman, G.~Imirzian, E.~Kaltofen.
\newblock A system for manipulating polynomials given by straight-line
  programs.
\newblock In {\em Proceedings of the fifth ACM Symposium on Symbolic and
  Algebraic Computation}, SYMSAC '86,  169--175. ACM, 1986.

\bibitem[GH01]{GH01}
M.~Giusti, J.~Heintz.
\newblock Kronecker's smart, little black boxes.
\newblock In {\em Foundations of Computational Mathematics, R. A. DeVore, A.
  Iserles, E. S{\"u}li eds.},  284 of {\em London Mathematical Society Lecture
  Note Series},  69--104. Cambridge University Press, Cambridge, 2001.

\bibitem[GHH{\etalchar{+}}97]{Giusti2}
M.~Giusti, K.~H\"{a}gele, J.~Heintz, J.~L. Monta{\~n}a, J.~E. Morais, L.~M.
  Pardo.
\newblock Lower bounds for diophantine approximation.
\newblock {\em Journal of Pure and Applied Algebra}, 117 277--317, 1997.

\bibitem[GHM{\etalchar{+}}98]{Giusti1}
M.~Giusti, J.~Heintz, J.E. Morais, J.~Morgenstern, L.M. Pardo.
\newblock Straight-line programs in geometric elimination theory.
\newblock {\em Journal of Pure and Applied Algebra}, 124 101--146, 1998.

\bibitem[GHMP95]{GHMP95}
M.~Giusti, J.~Heintz, J.~E. Morais, L.~M. Pardo.
\newblock When polynomial equation systems can be "solved" fast?
\newblock In {\em Applied Algebra, Algebraic Algorithms and Error Correcting
  Codes. Proc. of the 11th Intern. Conference AAECC. LNCS, 948},  205--231.
  Springer, 1995.

\bibitem[GHMP97]{Giusti3}
M.~Giusti, J.~Heintz, J.~E. Morais, L.~M. Pardo.
\newblock Le r\^ole des structures de donn\'ees dans les problemes
  d'\'elimination.
\newblock {\em Comptes Rendus Acad. Sci.}, Serie 1 (325) 1223--1228, 1997.

\bibitem[GHMS11]{GHMS09}
N.~Giménez, J.~Heintz, G.~Matera, P.~Solernó.
\newblock Lower complexity bounds for interpolation algorithms.
\newblock {\em Journal of Complexity}, 27 151--187, 2011.

\bibitem[GLS01]{GLS01}
M.~Giusti, G.~Lecerf, B.~Salvy.
\newblock A {G}r{\"{o}}bner {F}ree {A}lternative for {P}olynomial {S}ystem
  {S}olving.
\newblock {\em Journal of Complexity}, 17 154--211, 2001.

\bibitem[Hei89]{Hei89}
J.~Heintz.
\newblock On the computational complexity of polynomials and bilinear mappings.
  {A} survey.
\newblock {\em Proceedings 5th International Symposium on Applied Algebra,
  Algebraic Algorithms and Error Correcting Codes Springer LNCS}, 356 269--300,
  1989.

\bibitem[HKR11]{HKR11}
J.~Heintz, B.~Kuijpers, A.~{Rojas Paredes}.
\newblock Software engineering and complexity in effective algebraic geometry.
\newblock {\em Submited to J. of Complexity},
  \url{http://arxiv.org/PS_cache/arxiv/pdf/1110/1110.3030v2.pdf}, 2011.

\bibitem[HM93]{HeintzMorgenstern93}
J.~Heintz, J.~Morgenstern.
\newblock On the intrinsic complexity of elimination theory.
\newblock {\em Journal of Complexity}, 9 471--498, 1993.

\bibitem[HMW01]{HMW01}
J.~Heintz, G.~Matera, A.~Weissbein.
\newblock On the time--space complexity of geometric elimination procedures.
\newblock {\em Applicable Algebra in Engineering. Communication and Computing
  (AAECC Journal)}, 11 239--296, 2001.

\bibitem[HS81]{HeSie81}
J.~Heintz, M.~Sieveking.
\newblock Absolute primality of polynomials is decidable in random polynomial
  time in the number of variables.
\newblock In S.~Even, O.~Kariv, editors, {\em International Colloquium on
  Automata, Languages and Programming ICALP},  115 of {\em Lecture Notes in
  Computer Science},  16--28. Springer, 1981.

\bibitem[Ive73]{Iversen73}
B.~Iversen.
\newblock {\em Generic Local Structure of the Morphisms in Commutative
  Algebra}.
\newblock Springer-Verlag, Berlin, 1973.

\bibitem[Kal88]{Ka88}
E.~Kaltofen.
\newblock Greatest common divisors of polynomials given by straight--line
  programs.
\newblock {\em Journal of the ACM}, 35 (1) 231--264, 1988.

\bibitem[Kro82]{Kro1882}
L.~Kronecker.
\newblock {G}rundz{\"{u}}ge {E}iner arithmetischen {T}heorie der algebraischen
  {G}r{\"{o}}ssen ({F}undamentals of an arithmetic theory of algebraic
  quantities).
\newblock {\em J. Reine Angew. Math.}, 92 1--122, 1882.

\bibitem[Kun85]{Kunz85}
E.~Kunz.
\newblock {\em Introduction to commutative algebra and algebraic geometry}.
\newblock Birkh{\"{a}}user, Boston, 1985.

\bibitem[Lan93]{Lang93}
S.~Lang.
\newblock {\em Algebra}.
\newblock Addison-Wesley, Massachusetts, 1993.

\bibitem[Lec]{Lecerf}
G.~Lecerf.
\newblock Kronecker: a {M}agma package for polynomial system solving.
\newblock \url{http://lecerf.perso.math.cnrs.fr/software/kronecker/index.html}.

\bibitem[Lic90]{Lickteig90}
T.~M. Lickteig.
\newblock On semialgebraic decision complexity. {H}abilitationsschrift,
  {U}niversit{\"{a}}t {T}{\"{u}}bingen {TR}-90-052, {I}nt. {C}omp. {S}c.
  {I}nst., {B}erkeley, 1990.

\bibitem[Mac16]{M16}
F.~S. Macualay.
\newblock {\em Algebraic theory of modular systems}.
\newblock Tracts Math. Cambridge University Press, 1916.

\bibitem[Mor03]{Mora03}
T.~Mora.
\newblock {\em {S}{P}{E}{S} {I}: {T}he {K}ronecker-{D}uval {P}hilosophy}.
\newblock Cambridge University Press, 2003.

\bibitem[Mor05]{Mora05}
T.~Mora.
\newblock {\em {S}{P}{E}{S} {I}{I}: {M}acaulay's {P}aradigm and {G}roebner
  {T}echnology}.
\newblock Cambridge University Press, 2005.

\bibitem[Mum88]{Mumford88}
D.~Mumford.
\newblock {\em The red book of varieties and schemes},  1358.
\newblock Springer, Berlin Heidelberg, New York, 1. edition, 1988.

\bibitem[PS73]{PS73}
M.~S. Paterson, L.~J. Stockmeyer.
\newblock On the number of nonscalar multiplications necessary to evaluate
  polynomials.
\newblock {\em SIAM Journal on Computing}, 2 60--66, 1973.

\bibitem[Sha94]{Shafarevich94}
I.~R. Shafarevich.
\newblock {\em Basic algebraic geometry: Varieties in projective space.}
\newblock Springer, Berlin Heidelberg, New York, 1994.

\bibitem[SS95]{SS95}
M.~Shub, S.~Smale.
\newblock On the intractability of {H}ilbert's {N}ullstellensatz and an
  algebraic version of ``{N}{P}$\not=${P}?''.
\newblock {\em Duke Math. J.}, 81 47--54, 1995.

\bibitem[vdW50]{vdW50}
B.~L. van~der Waerden.
\newblock {\em Modern {A}lgebra {I}{I}}.
\newblock Ungar, New York, 1950.

\bibitem[ZS60]{ZaSa60}
O.~Zariski, P.~Samuel.
\newblock {\em Commutative algebra {II}}, ~39.
\newblock Springer, New York, 1960.

\end{thebibliography}

\newcommand{\etalchar}[1]{$^{#1}$}

%%%%%%%%%%%%%%%%%%%%%%%%%%%%%%%%%%%%%%%%%%%%%%%%%%%%%%%%%%%%%%%%%%%%%%%%%
%%%%%%%%%%%%%%%%%%%%%%  END OF DOCUMENT  %%%%%%%%%%%%%%%%%%%%%%%%%%%%%%%%
%%%%%%%%%%%%%%%%%%%%%%%%%%%%%%%%%%%%%%%%%%%%%%%%%%%%%%%%%%%%%%%%%%%%%%%%%

\end{document}